\documentclass[ alpha-refs]{wiley-article}

\usepackage{siunitx}
\usepackage{multirow}
\usepackage{amsmath}
\usepackage{subcaption}
\usepackage{booktabs}
\usepackage{diagbox}
\usepackage{float}
\usepackage{appendix}
\usepackage{natbib}
\usepackage[utf8]{inputenc}
\usepackage[english]{babel}
\usepackage{graphicx}
\usepackage{multicol}
\usepackage{setspace}
\usepackage{soul}

\newcommand*\circled[1]{\tikz[baseline=(char.base)]{
            \node[shape=circle,draw,inner sep=0.5pt] (char) {#1};}}

\begin{document}
\setlength{\parskip}{0pt}

\papertype{Research Article}

\title{Multi-step Reflection Principle and Barrier Options}


\author[1]{Hangsuck Lee}
\author[1]{Gaeun Lee}
\author[2]{Seongjoo Song}

\affil[1]{Department of Actuarial Science/Mathematics, Sungkyunkwan University, Myeongnyun 3(sam)ga, Jongno-gu,Seoul, 03063, Korea}
\affil[2]{Department of Statistics, Korea University, Anam-Ro 145, Seongbuk-Gu, Seoul 02841, Korea}

\corraddress{Seongjoo Song, Department of Statistics, Korea University, Anam-Ro 145, Seongbuk-Gu, Seoul 02841, Korea}
\corremail{sjsong@korea.ac.kr}

\runningauthor{Lee et al.}

\maketitle

\begin{abstract}

This paper examines a class of barrier options--\textit{multi-step barrier options}, which can have any finite number of barriers of any level. We obtain a general, explicit expression of option prices of this type under the Black-Scholes model. Multi-step barrier options are not only useful in that they can handle barriers of different levels and time steps, but can also approximate options with arbitrary barriers. Moreover, they can be embedded in financial products such as deposit insurances based on jump models with simple barriers. Along the way, we derive \textit{multi-step reflection principle}, which generalizes the reflection principle of Brownian motion.

\keywords{Brownian motion, \ reflection principle, \ multi-step reflection principle, \ Esscher transform, \ barrier option, \ multi-step barrier, \  icicles}
\noindent
{\bf Data Availability Statement:} Data sharing is not applicable to this article as no new data were created or analyzed in this study.
\end{abstract}

\section{INTRODUCTION}
Barrier options are a type of options whose payoff depends on whether the underlying asset price reaches a pre-determined barrier level. Ordinary barrier options have one horizontal barrier level throughout the entire lifetime, but \cite{lee2019a} proposed step barrier options, allowing three different piecewise constant barrier levels as in Figure \ref{Fig1}.  In this paper, we propose \textit{multi-step barrier options}, which generalize the step barrier options for cases with more flexible barrier structure. They can be further generalized by vertical branches of barriers, called \textit{icicles}, which enables an easy embedding of multi-step barrier options into equity-linked products. See Figure \ref{Fig2} for a multi-step barrier with more subperiods. For a simple representation, we leave out icicles in Figure \ref{Fig1} and Figure \ref{Fig2}.

\begin{figure}[H]
\centering
\includegraphics[width=0.6\textwidth]{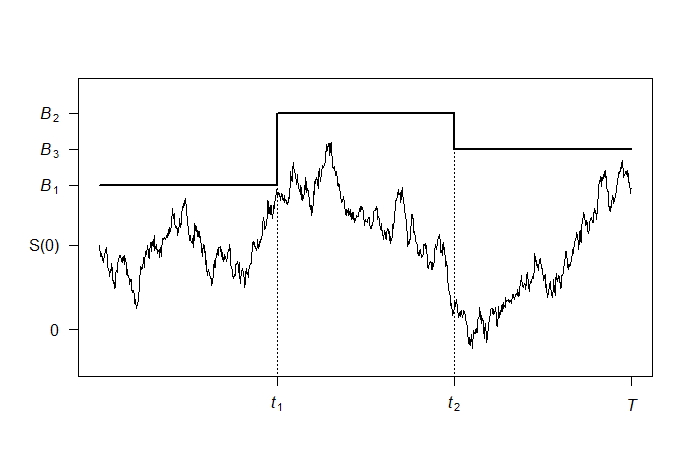}
\caption[step]{Three-step Barrier}
\label{Fig1}
\end{figure}

\begin{figure}[H]
\centering
\includegraphics[width=0.6\textwidth]{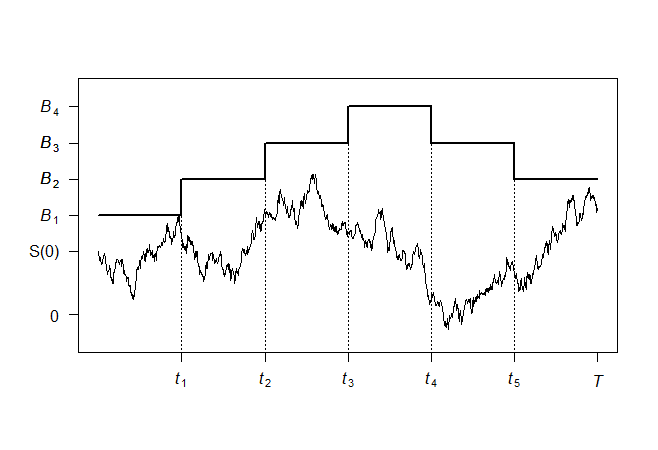}
\caption[mstep]{Multi-step Barrier}
\label{Fig2}
\end{figure}

At first sight, there may not seem to be much of a difference except the number of steps between step barriers proposed in \cite{lee2019a} and multi-step barriers proposed in this paper. However, what we obtain in this paper means more than just an increased number of steps. It offers a unified general expression for barrier option valuation formula with any finite number of steps. The formula includes the result of \cite{leeko2018}, \cite{lee2019a}, and \cite{lee2019b} as special cases. Partial barrier options with different barrier levels can also be handled with our formula. Besides, 
multi-step barrier options have the advantage of easily approximating options with arbitrary barriers, including curved barriers, with sufficiently many steps. Figure \ref{Fig3} shows examples of approximating curved barriers with multi-step barriers. Since the multi-step barrier probabilities in Theorem \ref{thm1} allow any finite number of steps and any barrier levels, more complicated barriers can be approximated by multi-step barriers with many steps. Subintervals of time do not have to be equally-spaced in the approximation.

\bigskip
\begin{figure}[h]
\centering
\begin{subfigure}{.6\textwidth}
\centering
\includegraphics[width=0.9\textwidth]{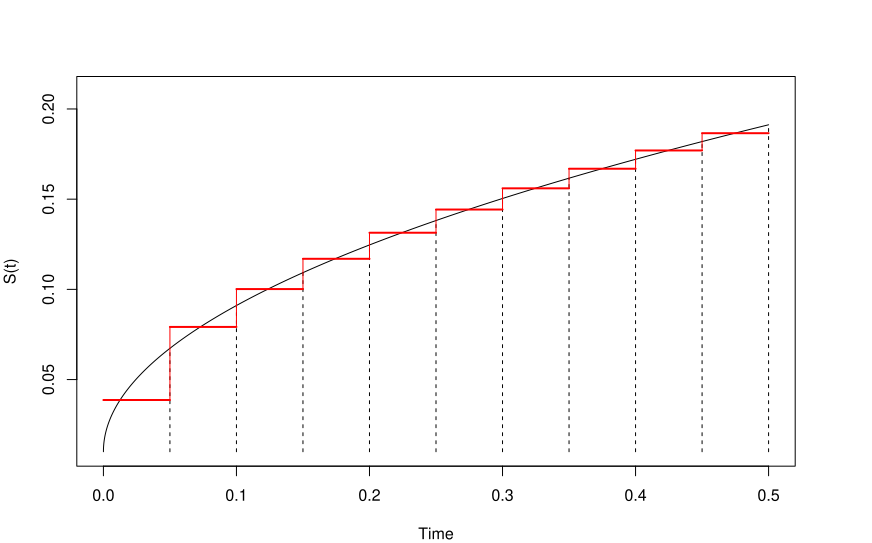}
\label{Curved1}
\end{subfigure}%

\begin{subfigure}{.6\textwidth}
\centering
\includegraphics[width=0.9\textwidth]{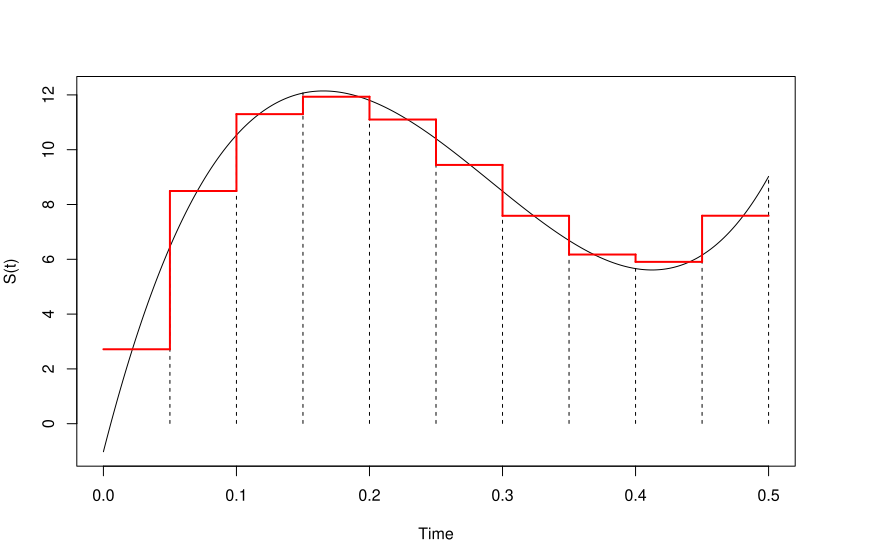}
\label{Curved2}
\end{subfigure}%

\caption[curved barriers]{Approximation of curved barriers with multi-step barriers}
\label{Fig3}
\end{figure}

Multi-step barrier options, while having flexible payoff structures, allow for explicit pricing formulas under the Black-Scholes model. To the purpose, the joint distribution of a Brownian motion at the end of subintervals of time and their partial maximums is to be obtained. When $\{X(t): t \geq 0\}$ is a Brownian motion with drift parameter $\mu$ and diffusion parameter $\sigma$ and $M(s,t)$ is the maximum of this Brownian motion during the time interval from $s$ to $t$, we obtain the joint distribution of $(X(t_1), \cdots, X(t_n))$ and partial maximums $(M(0, t_1), M(t_1, t_2), \cdots, M(t_{n-1}, t_n))$ in that Brownian motion hits step barriers at selected subintervals. One way of finding such joint distribution is the repeated application of reflection principle of the Brownian motion. Here, we generalize the reflection principle of the Brownian motion to allow multiple reflections and call the generalization as the \textit{multi-step reflection principle}. In the multi-step reflection, we use possibly different levels of reflecting barriers as opposed to \cite{lee2019b} which deals with identical barrier levels for every subinterval. This joint distribution plays a critical role in our derivation of the pricing formulas for multi-step barrier options and their variants. Through numerical examples, we explore the pricing formulas along with Monte Carlo simulated prices. We also demonstrate their applicability to the analysis of an investment with an arbitrary curved barrier in Examples \ref{ex4} and \ref{ex5}. Table \ref{Kunitomo8} compares the approximated prices by multi-step barriers with prices computed by the formula given in \cite{kunitomo1992}.

Note that our pricing methodology is completely probabilistic so that the mathematical background regarding partial differential equations (PDE) is not necessary. Instead, we exploit the method of Esscher transform and the factorization formula.

The Esscher transform is a time-honored actuarial tool, which has been applied to option pricing theory since \cite{gerbershiu1994}. Because it allows for efficient calculation of probability and expectation, complicated options embedded in equity-linked insurance contracts could be evaluated explicitly using it. 
See, for instance, \cite{tiong2000}, \cite{lee2003}, \cite{ngli2011}, \cite{gerberetal2012} and more recently, \cite{leeko2018}. As with previous literature, the Esscher transform 
and the factorization formula extremely simplify tedious calculations throughout this paper.
Esscher transform is also very popular in finding a risk-neutral measure in incomplete markets. Since \cite{gerbershiu1996} had justified the usage of the Esscher measure in an incomplete market, a considerable number of studies have been performed in this direction. See \cite{buhlmann1996}, \cite{schoutens2003}, \cite{chan1999}, \cite{jang2004}, and \cite{elliott2007} among others.

More complicated barrier options than ordinary ones are often found in the finance literature. To name a few, \cite{kunitomo1992} expressed the option price with double barriers in terms of infinite series and \cite{heynen1994} studied the pricing of barrier options where the barrier is restricted to a part of the options' lifetime. \cite{rogers1997} transformed the option pricing problem with smoothly moving double barriers to a fixed barrier option pricing problem which is solved numerically using a trinomial tree method. \cite{guillaume2010} investigated step double barrier options and \cite{hwang2015} dealt with an investment product whose payoff depends on certain pre-specified conditions. The fair premium of this product can be also represented with joint probabilities of univariate Brownian motion at a couple of time points and its running minimum as shown in \cite{wang2016}. \cite{leeko2018} introduced an icicled barrier option and their idea was extended in \cite{lee2019a} by setting the barrier as a piecewise constant one with icicles. \cite{lee2019b} considered the partial barriers at the same level with many icicles. 

The main purpose of this paper is two-fold: one is to introduce the multi-step reflection principle of Brownian motion and the other is to introduce the multi-step barrier options with or without icicles and find their explicit pricing formulas under the Black-Scholes model that do not involve an infinite series. The general expression for barrier option valuation formula with any finite number of steps is important in several ways. By allowing a more flexible barrier structure, the general pricing formula not only has a wider range of applications but can also approximate arbitrary barrier option prices with a desired precision. If we add icicles to the barrier structure, the range of application of multi-step barrier options will even be wider. Our formula can be applied to more general valuation of equity-linked products such as what is given in \cite{lee2019b} with more pliable payoff structure.

We may want to consider a more complicated model rather than the Black-Scholes model, but the Black-Scholes model would be a good starting point because the explicit pricing formulas even under the Black-Scholes model have never been obtained in the literature to our knowledge. And yet, nevertheless, multi-step barrier options can be applied to some jump models of the underlying asset. For example, when the jump events of the underlying asset are modeled by a binomial approach, we can transform the setting from a single barrier with jumps to a multi-step barrier without jumps. In other words, the differences in levels of multi-step barriers might be regarded as jump sizes of the underlying asset. In this direction, the multi-step barrier options can be embedded in other financial products such as deposit insurance schemes. 
For a deposit insurance that provides the guaranty when the underlying asset touches the regulatory barrier, we can take advantage of multi-step barriers to reflect the sudden movements of assets and liabilities in obtaining the explicit formula for its premium. Figure \ref{Fig4} illustrates this transformation from a jump model with a flat barrier to a continuous model with multi-step barriers.

\begin{figure}[h]
\centering
\begin{subfigure}{.6\textwidth}
\centering
\includegraphics[width=0.9\textwidth]{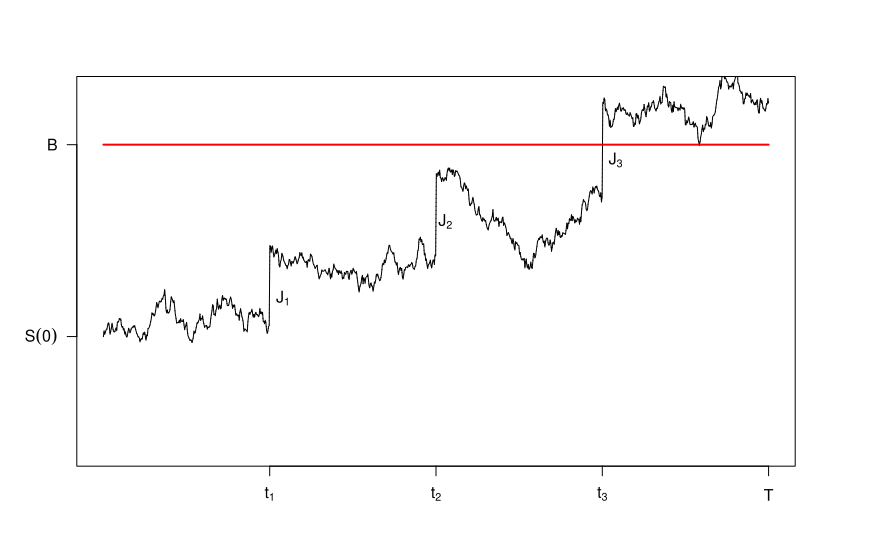}
\caption[jump_flat]{Sample path with jumps and a flat barrier}
\label{Jump1}
\end{subfigure}%

\begin{subfigure}{.6\textwidth}
\centering
\includegraphics[width=0.9\textwidth]{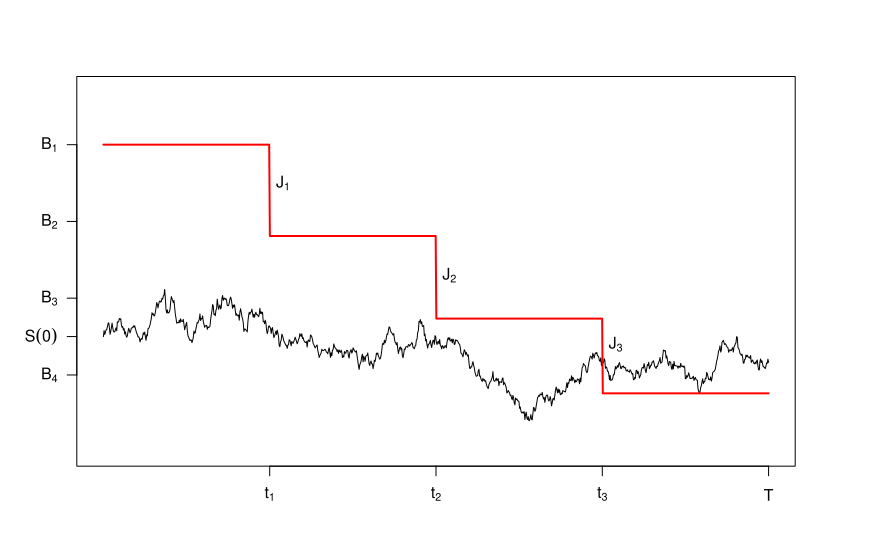}
\caption[nojump_step]{Sample path without jumps and a multi-step barrier}
\label{Jump2}
\end{subfigure}%

\caption[jump models]{Application of multi-step barriers to jump models; $J_1$, $J_2$, and $J_3$ are jump sizes}
\label{Fig4}
\end{figure}

The remainder of this paper is organized as follows. In Section 2, we formulate our mathematical framework and review some relevant materials including the Esscher transform. In Section 3, we introduce the multi-step reflection principle and derive the joint probabilities of our interest in terms of multivariate normal distributions. In Section 4, we find explicit pricing formulas for various icicled multi-step barrier options and explore option prices numerically at various input levels of barriers, interest rates, volatilities, strike prices, and maturities, along with Monte Carlo simulated prices for selected cases. In Section 5, we discuss examples of approximating the option prices with curved barriers using multi-step barrier option prices and in Section 6, we conclude the paper with some comments about the future research. 

\section{PRELIMINARIES}
\label{prelim}
Assume that the underlying asset price or index at time $t$ is described as
$$S(t) =S(0) e^{X(t)}, \qquad t \geq 0,$$
where $\{X(t): t \geq 0\}$ is a Brownian motion with the drift parameter $\mu$ and the diffusion parameter $\sigma$. 
With the maximum of the Brownian motion during the time interval from $s$ to $t$ denoted by
$$M(s,t)=\max \{ X(\tau): s \leq \tau \leq t \},$$ 
it is well known that the joint distribution of $X(t)$ and $M(0,t)$ is given by
\begin{equation}
\label{reflection}
Pr(X(t) \leq x, M(0,t) \leq m) = \Phi\left(\frac{x \wedge m-\mu t}{\sigma \sqrt{t}}\right)- e^{\frac{2 \mu}{\sigma^2}m} \Phi\left(\frac{x \wedge m -2m-\mu t}{\sigma\sqrt{t}}\right),
\end{equation}
for $m \geq 0$. Here $\Phi(\cdot)$ denotes the standard normal distribution function. The discussion paper of \cite{huangshiu2001} explains how to derive (\ref{reflection}) by using the method of Esscher transform. For an alternative derivation, we refer the readers to \cite{harrison1990}. It is obvious from (\ref{reflection}) that for $x \leq m$,
\begin{equation}
\label{reflection2}
Pr(X(t) \leq x, M(0,t) > m) =  e^{\frac{2 \mu}{\sigma^2}m} \Phi\left(\frac{x-2m-\mu t}{\sigma\sqrt{t}}\right)= e^{\frac{2 \mu}{\sigma^2}m} Pr(X(t) +2m \leq x).
\end{equation}

Let us briefly review the Esscher transform and the factorization formula. Esscher transform has been an extensively used tool in actuarial science since 
its development in \cite{gerbershiu1994}. For the usage of Esscher transform in equity-linked products, see, for example, \cite{tiong2000}, \cite{lee2003}, and \cite{gerberetal2012} among others. 
$\{ e^{h X(t)}/E(e^{hX(t)})\}$ is a positive martingale to define a new probability measure $Q$, which we call as Esscher measure of the parameter $h$ of the original probability measure $P$. 
For a random variable $Y$ that is a function of $\{X(t),  0 \leq t \leq T \}$, the expectation of $Y$ under $Q$ can be obtained as
\begin{equation}
\label{expyh}
E^Q(Y)=E^P\left(Y \frac{e^{h X(T)}}{E^P(e^{hX(T)})} \right).
\end{equation}
The expectation of $Y$ under the Esscher measure of the parameter $h$, $Q$, in (\ref{expyh}) will be denoted as $E(Y;h)$ in this paper. 
We can easily show that under $Q$, the Brownian motion with the drift, $\mu$ and the diffusion coefficient, $\sigma$ is changed to the Brownian motion with the drift, $\mu + h \sigma^2$ and the same diffusion coefficient. It allows us to find the risk-neutral measure under which the discounted stock price becomes a martingale. Esscher transform with the parameter $h^*$ such that $\mu+h^*\sigma^2 =r -\frac{1}{2} \sigma^2$ defines the risk-neutral measure to be used in the option pricing. 

In the pricing of barrier options, the following form of the factorization formula is useful. $I( \cdot)$ is the indicator function and $Pr(B;h)$ is the probability of the event $B$ under the Esscher measure of parameter $h$.
\begin{align*}
E(e^{X(T)}I(B);h^*)&=E^P(e^{X(T)}I(B) \frac{e^{h^* X(T)}}{E^P(e^{h^* X(T)})})\\
      &=\frac{E^P(e^{(h^*+1)X(T)})}{E^P(e^{h^*X(T)})} \frac{E^P(I(B)e^{(h^*+1)X(T)})}{E^P(e^{(h^*+1)X(T)})}\\
      &=E(e^{X(T)};h^*)E(I(B);h^*+1)\\
      &=e^{rT}Pr(B; h^*+1).
\end{align*}

\section{Multi-step Reflection Principle and Multi-step Barriers}\label{derivation}

In this section, we introduce the multi-step reflection principle and calculate joint probabilities that are useful in computing option prices with multi-step barriers. Before the mathematical details of the multi-step reflection principle, let us describe the way it works with a sample path of a Brownian motion with zero drift. Figures \ref{multirp:sub1}-\ref{multirp:sub3} show how the reflection behaves with varying barrier levels. Suppose we have the original path of a Brownian motion $\{ X^0(t): t \geq 0\}$ with zero drift and the diffusion coefficient $\sigma$ as in Figure \ref{multirp:sub1}. The superscript $0$ of $X$ is used to emphasize that the drift of the underlying Brownian motion is zero.  Among subintervals of $(t_0, t_1)$, $(t_1,t_2)$, and $(t_2, t_3)$, the first and the third subintervals have horizontal barriers with levels of $m_1$ and $m_3$, respectively. Also, there are icicles, $x_1 \leq m_1$, $x_2 \leq m_3$, and $x_3 \leq m_3$ at time points $t_1$, $t_2$, and $t_3$, respectively. The reason that $x_i$'s are called as icicles is obvious from the figure. 

\begin{figure}[H]
\centering
\includegraphics[width=0.7\textwidth]{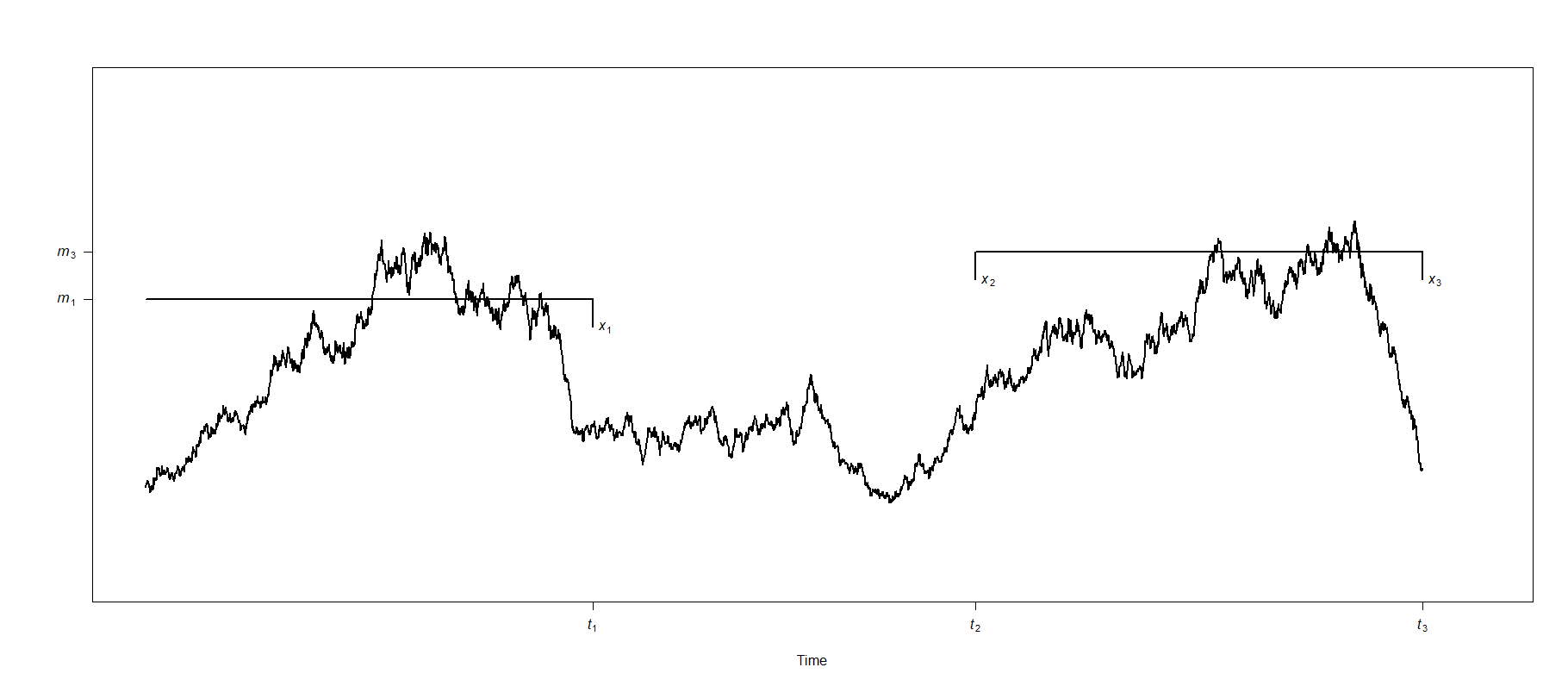}
\caption[Original]{Original path with partial barriers}
\label{multirp:sub1}
\end{figure}

To find the probability 
$$Pr(X^0(t_1) \leq x_1, X^0(t_2) \leq x_2, X^0(t_3) \leq x_3 , M^0(t_0, t_1) > m_1, M^0(t_2,t_3)>m_3),$$
we will use the ordinary reflection principle twice. Here, $M^0(s,t)$ is $\max \{X^0(\tau): \tau \in (s,t) \}$. When the path hits the barrier $m_1$ in $(t_0, t_1)$, the path is reflected around $m_1$. Since the reflected process follows the same Brownian motion as the original process by the ordinary reflection principle, we can easily see that
\begin{align*}
Pr(&X^0(t_1) \leq x_1, X^0(t_2) \leq x_2, X^0(t_3) \leq x_3 , M^0(t_0, t_1) > m_1, M^0(t_2,t_3)>m_3)\\
     &=Pr(2m_1 -X^0(t_1) \leq x_1, 2m_1 -X^0(t_2) \leq x_2, 2m_1-X^0(t_3) \leq x_3, \\
     &\qquad \qquad  \max_{t_2 < \tau <t_3} \{2m_1 -X^0(\tau) \} > m_3)\\
    &=Pr(X^0(t_1) \geq 2m_1-x_1, X^0(t_2) \geq 2m_1-x_2, X^0(t_3) \geq 2m_1-x_3, \\
    &\qquad \qquad \min_{t_2 < \tau <t_3} \{X^0(\tau)\} < 2m_1-m_3).
    \end{align*}

The once reflected path is shown in Figure \ref{multirp:sub2}. Note that the second barrier in $(t_2, t_3)$ becomes a lower barrier at $2m_1-m_3$ after the reflection. Moreover, the reflected icicles are stretched from the barrier upwards. 

When the once reflected path hits the lower barrier $2m_1-m_3$ in $(t_2, t_3)$, then the path is again reflected as shown in Figure \ref{multirp:sub3}. Then,
\begin{align*}
Pr(&X^0(t_1) \geq 2m_1-x_1, X^0(t_2) \geq 2m_1-x_2, X^0(t_3) \geq 2m_1-x_3, \\
    &\qquad \qquad \min_{t_2<\tau <t_3} \{X^0(\tau) \} < 2m_1-m_3)\\
    &=Pr(X^0(t_1) \geq 2m_1-x_1, X^0(t_2) \geq 2m_1-x_2, 2(2m_1-m_3)-X^0(t_3) \geq 2m_1 -x_3)\\
    &=Pr(X^0(t_1) \geq 2m_1-x_1, X^0(t_2) \geq 2m_1-x_2, -X^0(t_3)-2(m_3-m_1) \geq -x_3)\\
    &=Pr(-X^0(t_1)+2m_1 \leq x_1, -X^0(t_2)+2m_1 \leq x_2, X^0(t_3)+2(m_3-m_1) \leq x_3).
    \end{align*}    

Now, we generalize the above argument with arbitrary number of subintervals and horizontal barriers in the next theorem.
\begin{figure}[H]
\centering
\includegraphics[width=0.8\textwidth]{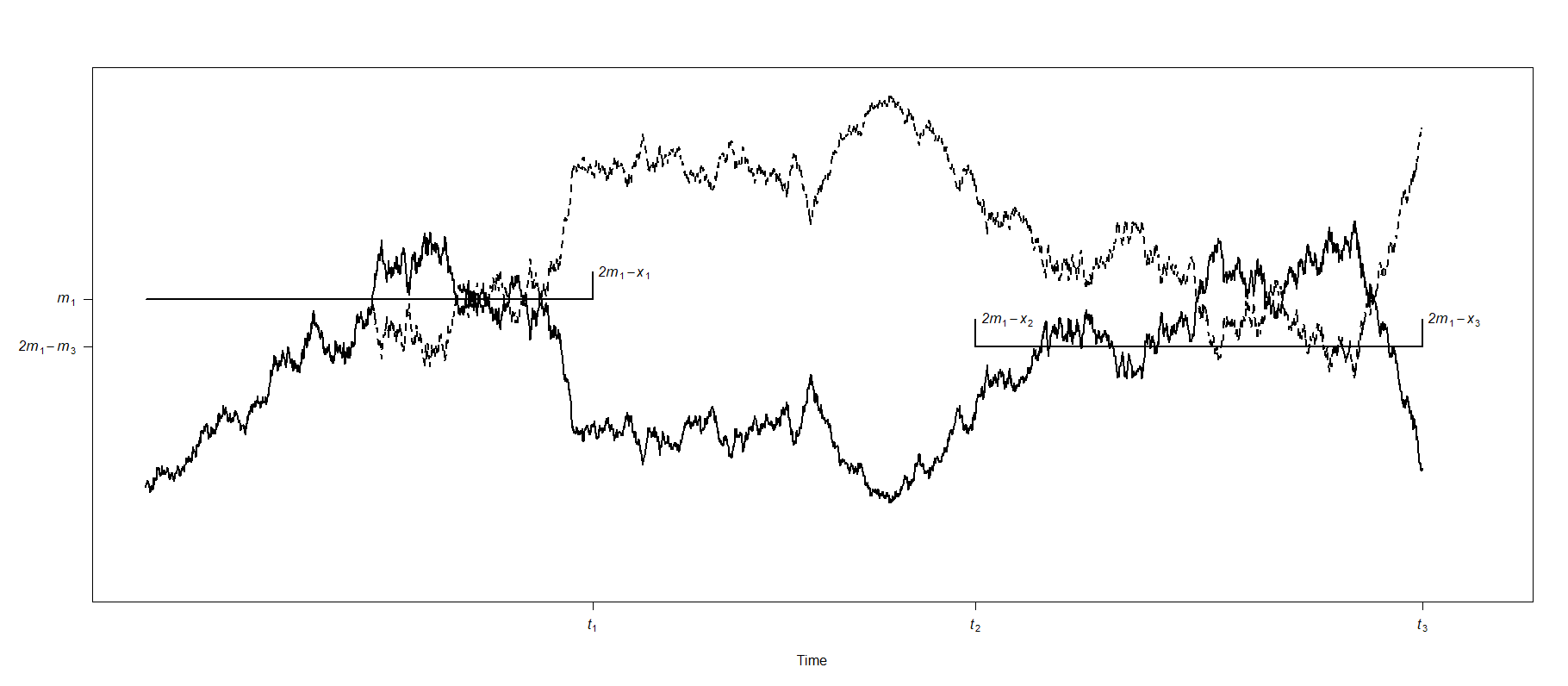}
\caption[Once Reflected]{Once Reflected path with partial barriers}
\label{multirp:sub2}
\end{figure}

\begin{figure}[H]
\centering
\includegraphics[width=0.8\textwidth]{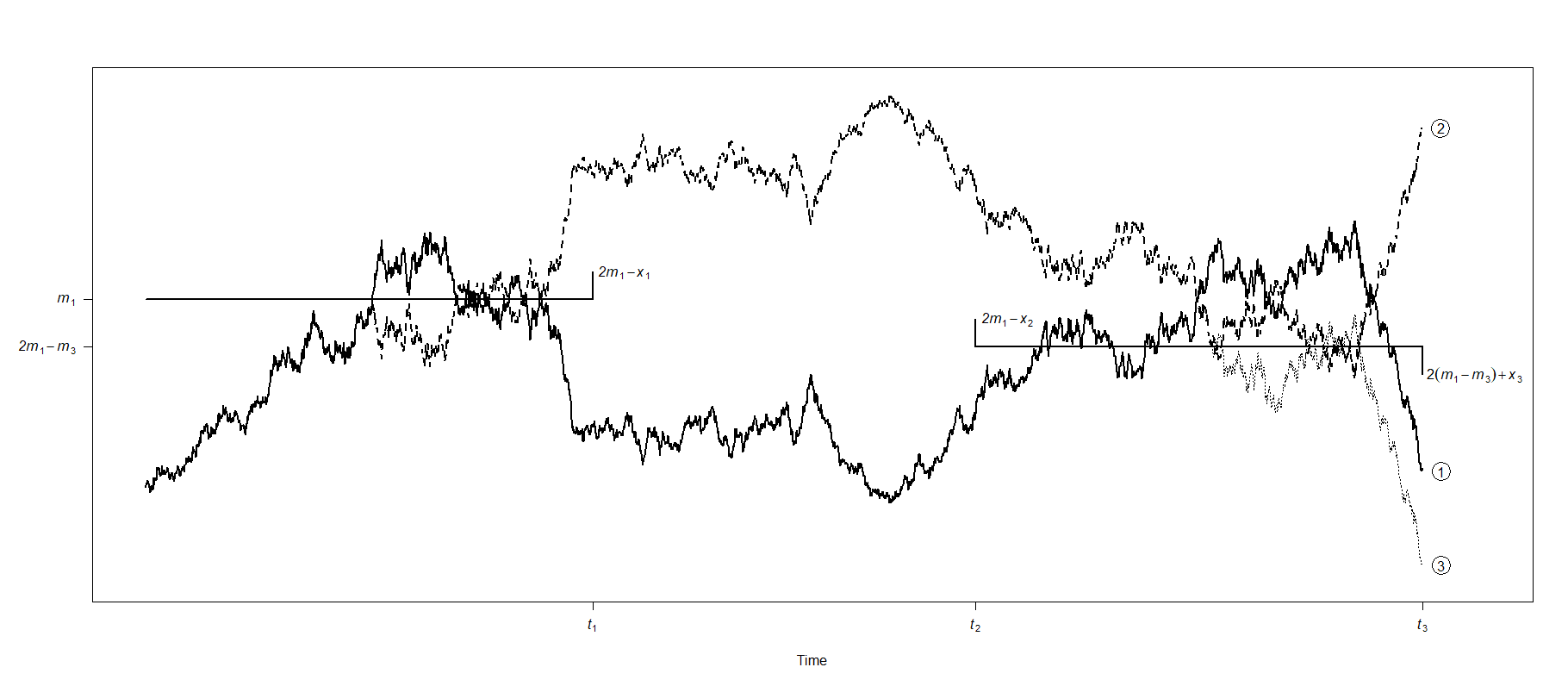}
\caption[Twice Reflected]{Twice Reflected path with partial barriers: \circled{1} is the original path, \circled{2} is the once reflected path, and \circled{3} is the twice reflected path}
\label{multirp:sub3}
\end{figure}

\bigskip
\begin{theorem}[Multi-step Reflection Principle]
\label{thm0}
Consider time points of $0=t_0 < t_1 < \cdots <t_{n-1}<t_n =  T$ and let $J$ be an arbitrary subset of $\{ 1, 2, \cdots, n\}$.
    
Let $\{X^0(t): 0 \leq t \leq T \}$ denote a Brownian motion with zero drift and diffusion coefficient $\sigma$ and $M^0(s,t)$ denote the maximum of $X^0$ in the time interval from $s$ to $t$. Then
\begin{equation}
\label{eqn3_30}
\begin{split}
Pr(\cap_{i=1}^n \{ X^0(t_i) \leq x_i \} , &\cap_{\{k \in J\}} \{M^0(t_{k-1}, t_k) > m_k\})\\
                   &= Pr( \cap_{i=1}^n \{s_i X^0(t_i) +2m[i] \leq x_i \}),
\end{split}
\end{equation}
where $\{m[i]: i=0, \cdots, n\}$ is a sequence such that $m[0] = 0$ and
\begin{equation}
\label{eqn3_1}
m[i] := (m_i - m[i-1])I(i  \in J) + m[i-1]I(i \notin J),
\end{equation}
and
\begin{equation}  
\label{eqn3_2}
s_i = \begin{cases} 1, & \text{ if the number of elements that are greater than $i$ in $J$ is even,}\\

-1, & \text{otherwise,}
\end{cases}
\end{equation}
for $i=1, \cdots, n$. In other words, $s_i=1$ if the cardinality of the set $\{k>i : k \in J\}$ is even. Note that $s_n = 1$. Here, we set $m_1 \geq 0$, $x_{k-1} \leq m_k$ and $x_{k} \leq m_k$ for $k \in J$.
\end{theorem}

\begin{proof}
See Appendix \ref{proofthm0}.
\end{proof}

Now, we apply the multi-step reflection principle to a Brownian motion with nonzero drift in Theorem \ref{thm1}. The result of Theorem \ref{thm1} generalizes ``the generalized reflection princple'' of Brownian motion by \cite{lee2019b} even further by allowing the horizontal barrier levels different from subinterval to subinterval. 

\begin{theorem}
\label{thm1}  
Consider time points of $0=t_0 < t_1 < \cdots <t_{n-1}<t_n = T$ and let $\{X(t): 0 \leq t \leq T \}$ denote a Brownian motion with drift $\mu$ and diffusion coefficient $\sigma$. $M(s,t)$ denotes the maximum of $X$ in the time interval from $s$ to $t$ for $0<s<t \leq T$. Then 
\begin{equation}
\label{eqn3_3}
\begin{split}
Pr(\cap_{i=1}^n \{ X(t_i) \leq x_i \} , &\cap_{\{i \in J\}} \{M(t_{i-1}, t_i) > m_i\})\\
                   &= e^{\frac{2\mu}{\sigma^2} m[n]}  Pr( \cap_{i=1}^n \{s_i X(t_i) +2m[i] \leq x_i \}),
\end{split}
\end{equation}
where $J$, $m[i]$, and $s_i$ are defined as in Theorem \ref{thm0}. 
We also assume $m_1 \geq 0$, $x_{i-1} \leq m_i$ and $x_{i} \leq m_i$ for $i \in J$.
\end{theorem}
\begin{proof} See Appendix \ref{proofthm1}. \end{proof}

\begin{example} 
Let $\{X(t): 0 \leq t \leq T \}$ denote a Brownian motion with drift $\mu$ and diffusion coefficient $\sigma$. $M(s, t)$ is the maximum of $X$ in the interval from $s$ to $t$ for $0<s<t \leq T$.
Here we apply (\ref{eqn3_3}) to the case of $n=2$ with $0<t_1<t_2=T$. 
When $J=\{1\}$, for $m_1 \geq 0$ and $x_1 \leq m_1$,
\begin{align*} %
Pr(X(t_1) & \leq x_1, X(t_2) \leq x_2, M(0,t_1) > m_1)\\
             &=e^{\frac{2 \mu}{\sigma^2}m_1}Pr(X(t_1)+2m_1 \leq x_1, X(t_2)+2m_1 \leq x_2)\\
             &=e^{\frac{2 \mu}{\sigma^2}m_1} \Phi_2 \left(\frac{x_1-2m_1 -\mu t_1}{\sigma \sqrt{t_1}}, \frac{x_2-2m_1-\mu t_2}{\sigma \sqrt{t_2}}; \sqrt{\frac{t_1}{t_2}} \right)
\end{align*} %
and when $J=\{2\}$, for $x_1 \leq m_2$, and $x_2 \leq m_2$,
\begin{align*}
Pr(X(t_1) & \leq x_1, X(t_2) \leq x_2, M(t_1, t_2) > m_2)\\
	&=e^{\frac{2 \mu}{\sigma^2}m_2} Pr(-X(t_1) \leq x_1, X(t_2)+2m_2 \leq x_2) \\
	&=e^{\frac{2 \mu}{\sigma^2}m_2} \Phi_2 \left(\frac{x_1+\mu t_1}{\sigma \sqrt{t_1}}, \frac{x_2-2m_2-\mu t_2}{\sigma \sqrt{t_2}}; -\sqrt{\frac{t_1}{t_2}} \right).
\end{align*}
$\Phi_2(\cdot, \cdot ;\rho)$ is the cumulative distribution function of the 2-dimensional standard normal distribution with correlation coefficient $\rho$. 
The above results correspond to what appeared in \cite{leeko2018}. 
Also, for $J=\{1,2\}$, $m_1 \geq 0$, $x_1 \leq m_1 \wedge m_2$, and $x_2 \leq m_2$, we obtain
\begin{align*}
Pr(X(t_1)  &\leq x_1, X(t_2)  \leq x_2, M(0,t_1) > m_1, M(t_1, t_2)>m_2)\\
              &=e^{\frac{2 \mu}{\sigma^2}(m_2-m_1)} Pr(-X(t_1) +2m_1 \leq x_1, X(t_2)+2(m_2-m_1) \leq x_2) \\
              &=e^{\frac{2 \mu}{\sigma^2}(m_2-m_1)} \Phi_2 \left(\frac{x_1-2m_1 +\mu t_1}{\sigma \sqrt{t_1}}, \frac{x_2-2(m_2-m_1)-\mu t_2}{\sigma \sqrt{t_2}}; -\sqrt{\frac{t_1}{t_2}} \right). 
\end{align*}
\label{ex1}
\end{example}

\begin{example}
Applying (\ref{eqn3_3}) to the case of $n=3$, we obtain the following formulas. 
$\Phi_3(\cdot, \cdot, \cdot ;\rho_{12}, \rho_{13}, \rho_{23})$ given below is the cumulative distribution function of the 3-dimensional standard normal distribution with correlation coefficients, $\rho_{12}$, $\rho_{13}$, and $\rho_{23}$.

When $J=\{ 2\}$, for $x_1 \leq m_2$, and $x_2 \leq m_2$,
\begin{align*}
Pr(&X(t_1) \leq x_1, X(t_2) \leq x_2, X(t_3) \leq x_3, M(t_1, t_2) > m_2)\\
 		&=e^{\frac{2 \mu m_2}{\sigma^2}} Pr(-X(t_1) \leq x_1, X(t_2)+2m_2 \leq x_2, X(t_3)+2m_2 \leq x_3)\\
           &=e^{\frac{2 \mu m_2}{\sigma^2}} \Phi_3 \left( \frac{x_1+\mu t_1}{\sigma \sqrt{t_1}}, \frac{x_2-2m_2-\mu t_2}{\sigma \sqrt{t_2}}, 
                   \frac{x_3 -2m_2-\mu t_3}{\sigma \sqrt{t_3}}; -\sqrt{\frac{t_1}{t_2}}, -\sqrt{\frac{t_1}{t_3}}, \sqrt{\frac{t_2}{t_3}} \right),
\end{align*}
when $J=\{2,3\}$, for $x_1 \leq m_2$, $x_2 \leq m_2 \wedge m_3$, and $x_3 \leq m_3$,
\begin{align*}
Pr(&X(t_1) \leq x_1, X(t_2) \leq x_2, X(t_3) \leq x_3, M(t_1,t_2) >m_2, M(t_2, t_3) > m_3)\\
		&=e^{\frac{2 \mu}{\sigma^2}(m_3-m_2)} Pr(X(t_1) \leq x_1, -X(t_2)+2m_2 \leq x_2, X(t_3)+2(m_3-m_2) \leq x_3)\\
           &=e^{\frac{2 \mu}{\sigma^2}(m_3-m_2)} \Phi_3 \left( \frac{x_1 -\mu t_1}{\sigma \sqrt{t_1}}, \frac{x_2 -2m_2+\mu t_2}{\sigma \sqrt{t_2}}, \right.\\
          &\hspace{3cm} \left. \frac{x_3 -2(m_3-m_2)-\mu t_3}{\sigma \sqrt{t_3}}; -\sqrt{\frac{t_1}{t_2}}, \sqrt{\frac{t_1}{t_3}}, -\sqrt{\frac{t_2}{t_3}} \right),
\end{align*}
and when $J=\{1,2,3\}$, for $m_1 \geq 0$, $x_1 \leq m_1 \wedge m_2$, $x_2 \leq m_2 \wedge m_3$, and $x_3 \leq m_3$,
\begin{align*}
Pr(&X(t_1) \leq x_1, X(t_2) \leq x_2, X(t_3) \leq x_3, M(0, t_1) > m_1, M(t_1, t_2) > m_2, M(t_2, t_3) > m_3 )\\
                    &=e^{\frac{2 \mu}{\sigma^2}(m_3-m_2+m_1)} Pr(X(t_1) +2m_1 \leq x_1,\\
                    &\qquad \qquad \qquad -X(t_2)+ 2(m_2 -m_1) \leq x_2, X(t_3) +2(m_3-m_2+m_1) \leq x_3)\\
                    &= e^{\frac{2 \mu}{\sigma^2}(m_3-m_2+m_1)} \Phi_3 \left( \frac{x_1-2m_1 -\mu t_1}{\sigma \sqrt{t_1}}, \frac{x_2-2(m_2-m_1)+\mu t_2}{\sigma \sqrt{t_2}}, \right.\\
                    &\hspace{3.5cm} \left. \frac{x_3 -2(m_3-m_2 +m_1)-\mu t_3}{\sigma \sqrt{t_3}}; -\sqrt{\frac{t_1}{t_2}}, \sqrt{\frac{t_1}{t_3}}, -\sqrt{\frac{t_2}{t_3}} \right).
\end{align*}
These formulas coincide with what \cite{lee2019a} has calculated.
\label{ex2}
\end{example}
\newpage

\begin{example}
Applying (\ref{eqn3_3}) to the case of $m_i=m$ for all $i$, we obtain 
\begin{align*}
Pr(\cap_{i=1}^n  \{X(t_i) \leq x_i\}, & \cap_{\{i \in J\}} \{M(t_{i-1}, t_i) > m\})\\
                   &= e^{\frac{2\mu}{\sigma^2} m[n]}  Pr(\cap_{i=1}^n \{s_i X(t_i) +2m[i] \leq x_i\}),
\end{align*}
which is the same as Theorem 1 in \cite{lee2019b}. It is easy to see that $m[k]$ and $s_k$ in this paper are the same as $m_k$ and $a_k$ for $k=1, \cdots, n$ in \cite{lee2019b}, respectively. 
\label{ex3}
\end{example}

The following proposition is useful in computing option prices with multi-step barriers. The proof is trivial by the inclusion-exclusion identity of the probability.
\medskip
\begin{proposition}\label{prop1}
Let us consider an arbitrary subset, $I$, of $\{1, \cdots, n\}$ and let $\{X(t): 0 \leq t \leq T \}$ denote a Brownian motion with drift $\mu$ and diffusion coefficient $\sigma$. Then with $0=t_0 < t_1 < \cdots <t_{n-1}<t_n =  T$ and $m_1 \geq 0$,
\begin{equation}
\label{prop1eqn1}
\begin{split}
Pr(& \cap_{i=1}^n \{X(t_i) \leq x_i\}, \cap_{i \in I} \{M(t_{i-1}, t_i) \leq m_i\})\\
	&=Pr( \cap_{i=1}^n \{X(t_i) \leq x_i\}) - Pr(\cap_{i=1}^n \{X(t_i) \leq x_i\}, \cup_{i \in I} \{ M(t_{i-1}, t_i) > m_i \} )\\
	&=\sum_{J \subset I} (-1)^{|J|} Pr(\cap_{i=1}^n \{X(t_i) \leq x_i\}, \cap_{j \in J} \{ M(t_{j-1}, t_j) > m_j \} )\\
	&=\sum_{J \subset I} (-1)^{|J|} e^{\frac{2\mu}{\sigma^2} m^J[n]}Pr(\cap_{i=1}^n \{s^J_i X(t_i)+2m^J[i] \leq x_i \}). 
\end{split}
\end{equation}
Note that $J$ is any subset of $I$ and $|J|$ is the cardinality of the set $J$. $m^J[i]$ and $s^J_i$ are defined as (\ref{eqn3_1}) and (\ref{eqn3_2}) and we put the superscript $J$ here to emphasize the fact that their values are dependent on the index set $J$.
\end{proposition}

In the following corollary, we obtain the probability regarding the case without icicles by setting $x_i=m_i \wedge m_{i+1}$ for $i=1, \cdots, n-1$ and $x_n=m_n$.
\begin{corollary}
\label{coro1}
We consider the case where every subinterval has a horizontal barrier and let $\{X(t): 0 \leq t \leq T \}$ denote a Brownian motion with drift $\mu$ and diffusion coefficient $\sigma$.
Then for $m_1 \geq 0$,
\begin{equation}
\label{coro1eqn1}
\begin{split}
&Pr(\cap_{i=1}^n \{M(t_{i-1}, t_i) < m_i\})\\
	&=Pr(X(t_1) \leq m_1 \wedge m_2, \cdots, X(t_{n-1}) \leq m_{n-1} \wedge m_n, X(t_n) \leq m_n, \cap_{i=1}^n \{M(t_{i-1}, t_i) < m_i\})\\
	&=\sum_{J \subset \{1, \cdots, n\}} (-1)^{|J|} Pr(X(t_1) \leq m_1 \wedge m_2 , \cdots, X(t_{n-1}) \leq m_{n-1} \wedge m_n, X(t_n) \leq m_n, \\
	&\hspace{3cm} \cap_{j \in J} \{ M(t_{j-1}, t_j) > m_j \} ). 
\end{split}
\end{equation}
\end{corollary}

In the following remark, we observe the hierarchical structure contained in Theorem \ref{thm1} in terms of $J$, the set of intervals with partial barriers. 
\begin{remark}
\label{rem1}
Let $J  \subset J^*$ be index sets that indicate which time intervals have horizontal barriers. Then
\begin{equation}
\label{rem1eqn1}
\begin{split}
Pr(&\cap_{i=1}^n \{X(t_i) \leq x_i\}, \cap_{\{i \in J^*\}} \{M(t_{i-1}, t_i) > m_i\})\\
     & \leq Pr(\cap_{i=1}^n \{X(t_i) \leq x_i\}, \cap_{\{i \in J\}} \{M(t_{i-1}, t_i) > m_i\}).
\end{split}
\end{equation}
This seemingly trivial relationship can be useful in numerical computations. Suppose there are $n$ subintervals and $n$ horizontal barriers, and we want to calculate the probability, $Pr(\cap_{i=1}^n \{X(t_i) \leq x_i\}, \cap_{i=1}^n \{M(t_{i-1}, t_i) \leq m_i\})$. In order to calculate the probability of interest, we need to find $2^n$ probabilities in the form of $Pr(\cap_{i=1}^n \{X(t_i) \leq x_i\}, \cap_{j \in J} \{M(t_{j-1}, t_j) > m_j\})$ as in Proposition \ref{prop1}. Using the hierarchical relationship among such probabilities, if we find a ``negligible'' probability in a certain step, we can ignore the probability after that step in the hierarchical structure. As an example, let us consider a case of $n=3$ where we can think of 8 different sets of $J$'s; $\varnothing, \{1\}, \{2\}, \{3\}, \{1,2\}, \cdots, \{1,2,3\}$. Figure \ref{hier1} shows the hierarchical structure of (\ref{rem1eqn1}) according to the cardinality of $J$. For instance, (\ref{rem1eqn1}) with $J=\{1,2\}$ and $J^*=\{1,2,3\}$ is $0.0084 > 0.0008$. If $J$ provides a negligible probability, then $J^* \supset J$, appeared under $J$ in the figure, does not have to be considered at all. 
Here we assume $\mu=0.01$, $\sigma=0.2$, $t_1=2/12$, $t_2=4/12$, $t_3=T=6/12$, $m_1=\log 1.1$, $m_2=\log 1.2$, and $m_3=\log 1.3$. The probabilities in the boxes of Figure \ref{hier1} are obtained by (\ref{coro1eqn1}) since we assume no icicles.

\end{remark}

\begin{figure}[H]
\centering
\includegraphics[width=0.55\textwidth]{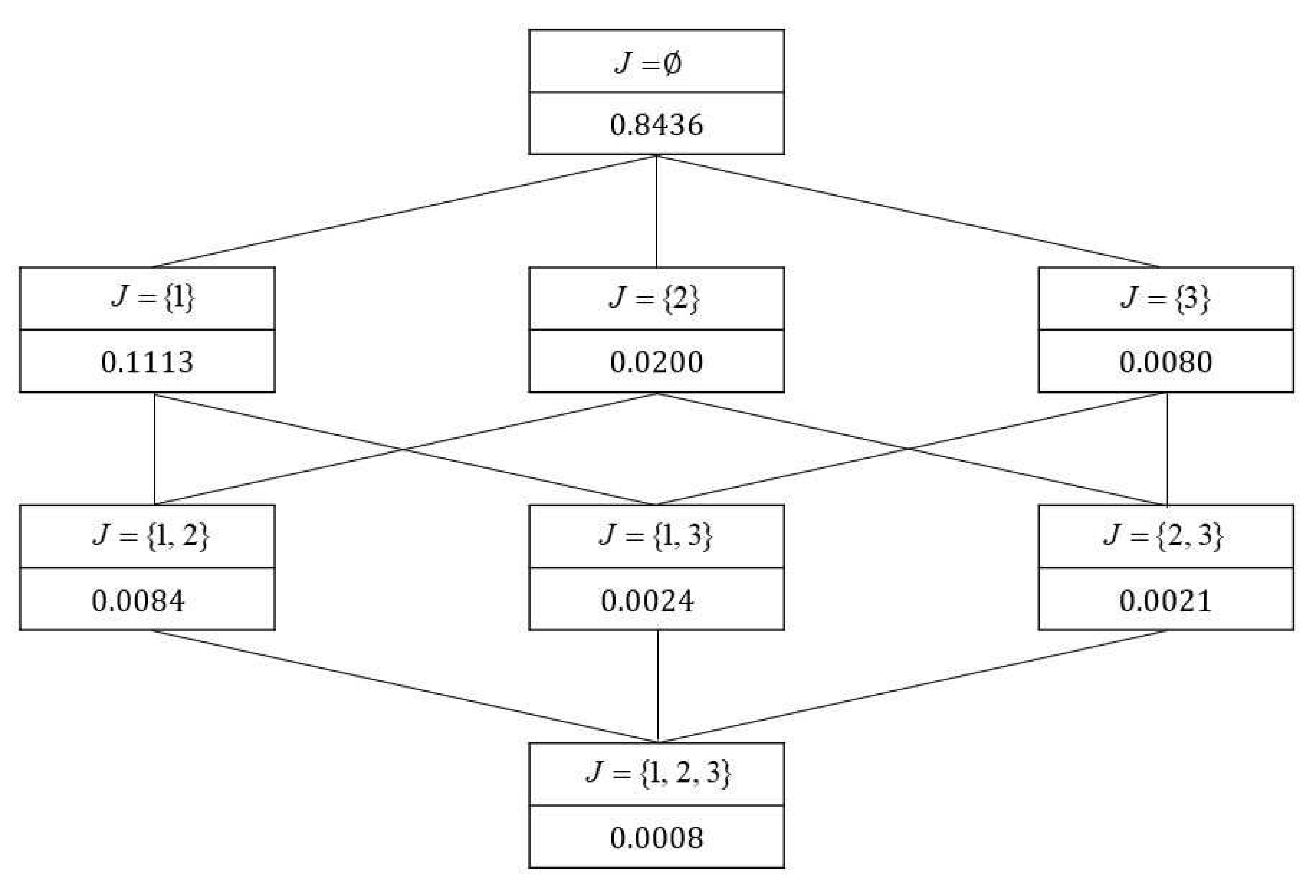}
\caption[hierarchical structure]{Hierarchical structure in Remark \ref{rem1} with $n=3$}
\label{hier1}
\end{figure}

\section{Pricing multi-step barrier options}\label{pricing}

Let us denote by $T$ the option maturity, and consider time points of $0 =t_0 < t_1 < \cdots <t_n = T$. The horizontal and icicled barrier levels in the interval $[t_{i-1}, t_i]$ are denoted as $B_i$ and $L_i$, respectively, for $i=1, \cdots, n$.
$I=\{i_1, i_2, \cdots, i_g\}$ is defined as in Proposition \ref{prop1} and would often be set as $\{1, 2, \cdots, n\}$ in practice. 
To simplify the triggering events, we define 
$$A_u:=\{S(t_1) \leq L_1 , \cdots, S(t_n) \leq L_n, \bigcap_{i \in I} \{\max \{ S(\tau): t_{i-1} \leq \tau \leq t_i \} \leq B_i\} \}$$
and 
$$A_d:=\{S(t_1) \geq L_1 , \cdots, S(t_n) \geq L_n, \bigcap_{i \in I} \{ \min \{ S(\tau): t_{i-1} \leq \tau \leq t_i \} \geq B_i \} \}.$$ Or, with $x_i = \ln(L_i /S(0))$ and $m_i=\ln(B_i/S(0))$, they can be equivalently rewritten as 
$$A_u:=\{X(t_1) \leq x_1 , \cdots, X(t_n) \leq x_n, \bigcap_{i \in I} \{ M(t_{i-1}, t_i) \leq m_i \} \}$$
and 
$$A_d:=\{X(t_1) \geq x_1 , \cdots, X(t_n) \geq x_n, \bigcap_{i \in I} \{m(t_{i-1}, t_i)  \geq m_i \}\}$$
where $m(s,t)=\min \{ X(\tau): s \leq \tau \leq t\}.$
For the strike price $K$, icicled multi-step barrier options and their corresponding payoff structures are summarized in Table \ref{payofftable}.  $A_u^c$ and $A_d^c$ are complements of the sets $A_u$ and $A_d$, respectively.

\begin{table}[hbt!]
\small
\centering
\caption{Icicled multi-step barrier options and their payoffs}
\label{payofftable}
\begin{tabular}{|c|c||c|c|}
\hiderowcolors
\hline
 Option type & Payoff & Option type & Payoff\\
 \hline
 Up and Out Put (UOP) & $(K-S(T))^+ I(A_u)$ &Up and In Put (UIP) & $(K-S(T))^+ I({A_u^c})$\\
 Up and Out Call (UOC) &$(S(T)-K)^+ I(A_u)$ & Up and In Call (UIC) & $(S(T)-K)^+ I(A_u^c)$\\
 \hline
  Down and Out Put (DOP) & $(K-S(T))^+ I(A_d)$ & Down and In Put (DIP) & $(K-S(T))^+ I({A_d^c})$\\
 Down and Out Call (DOC) &$(S(T)-K)^+ I(A_d)$ & Down and In Call (DIC) & $(S(T)-K)^+ I(A_d^c)$\\
 \hline
\end{tabular}
\end{table}

To obtain the option prices in a compact form, we write the probabilities associated with the events $A_u$ and $A_d$ as
\begin{align*}
PA_u(\mu, x_1, \cdots, x_n, m_{i_1}, \cdots, m_{i_g})&= Pr(X(t_1) \leq x_1 , \cdots, X(t_n) \leq x_n, \bigcap_{i \in I} \{ M(t_{i-1}, t_i) \leq m_i \} )\\
  &=\sum_{J \subset I} (-1)^{|J|} e^{\frac{2\mu}{\sigma^2}m^J[n]} Pr(\cap_{i=1}^n \{s^J_i X(t_i) +2m^J[i] \leq x_i\})
\end{align*}
and 
$$PA_d(\mu, x_1, \cdots, x_n, m_{i_1}, \cdots, m_{i_g})= Pr(X(t_1) \geq x_1, \cdots, X(t_n) \geq x_n, \bigcap_{i \in I} \{ m(t_{i-1}, t_i) \geq m_i \} )$$
respectively. 
Here the drift $\mu$ is explicitly specified with the functions $PA_u$ and $PA_d$ for the later measure change. Note that $PA_u$ is exactly the same as the probability given in Proposition \ref{prop1} and $PA_d$ is easily obtained through 
$$PA_d(\mu, x_1, \cdots, x_n, m_{i_1}, \cdots, m_{i_g})=PA_u(-\mu, -x_1, \cdots, -x_n, -m_{i_1}, \cdots, -m_{i_g}).$$

By the fundamental theorem of asset pricing and the factorization formula, the option pricing formulas are immediately obtained using the expectation of the discounted payoffs under the Esscher measure of parameter ${h}^{*}$ given in Section \ref{prelim}. For instance, the up-and-out icicled multi-step barrier put (UOP) price at time 0 is calculated as    
\begin{align*}
E & \left[e^{-rT} (K-S(T))^+ I(A_u);h^* \right]\\
		&= Ke^{-rT}Pr(A_u \cap (S(T)<K);h^*)-S(0)Pr(A_u \cap (S(T)<K); h^*+1)\\
		&=Ke^{-rT}PA_u \left(r-\frac{1}{2} \sigma^2, x_1, \cdots, x_n \wedge k, m_{i_1}, \cdots, m_{i_g}\right)\\
                &\qquad -S(0)PA_u\left(r+\frac{1}{2} \sigma^2, x_1, \cdots, x_n \wedge k, m_{i_1}, \cdots, m_{i_g}\right),
\end{align*}
where $k=\ln(K/S(0))$ and $a \wedge b = \min \{a, b\}$. 
Similarly, one may find the other option prices in Tables \ref{uppricetable} and \ref{downpricetable}. 
When constructing the tables, one needs to be cautious about the range of $K$.
For instance, if $K \geq L_n$, the UOC payoff would always become zero so that it is unnecessary to consider the case.   
Note that similar considerations should be taken for UIC, DOP and DIP prices. Setting $L_i = B_i \wedge B_{i+1}$, or equivalently, $x_i=m_i \wedge m_{i+1}$ for $i=1, \cdots, n-1$ and $L_n=B_n$ or $x_n=m_n$ would give rise to the multi-step barrier option prices without icicles. Corollary \ref{coro1} is used to find such prices. Also, setting $m_i=m$ for all $i$ gives the result of \cite{lee2019b}. 

\begin{example}\label{exkorn523}
Consider a down-and-out call option with ordinary barrier at $B_1$. We assume that the option expires in $T$, the strike price is $K$, and the interest rate is $r$. Then the ordinary down-and-out call option price is calculated using our method as follows.
\begin{align}
E & \left[e^{-rT} (S(T)-K)^+ I(\min_{0 \leq t\leq T} S(t) > B_1);h^* \right] \nonumber \\
   &=E\left[e^{-rT} (S(T)-K) I(S(T)>K, \min_{0 \leq t\leq T} S(t) > B_1);h^* \right] \nonumber \\
		&= e^{-rT}S(0)E\left[e^{X(T)}I(X(T)>k, m(0,T) >m_1); h^*\right]-Ke^{-rT}Pr(X(T)>k, m(0,T)>m_1;h^*) \nonumber \\
		&=S(0)Pr(X(T)>k, m(0,T) >m_1; h^*+1)-Ke^{-rT}Pr(X(T)>k, m(0,T)>m_1;h^*),\label{ordinary1}
\end{align}
where $k=\ln(K/S(0))$ and $m_1=\ln(B_1/S(0))$. Also,
\begin{align}
Pr(&X(T)>k, m(0,T) >m_1; h^*)=Pr(-X(T)<-k, -m(0,T) <-m_1; h^*) \nonumber \\
   &=Pr(Y(T)<-k, M_Y(0,T) <-m_1;h^*)\nonumber \\
   &=Pr(Y(T)<(-k) \wedge (-m_1), M_Y(0,T) <-m_1; h^*) \nonumber \\
   &=Pr(Y(T)<(-k) \wedge (-m_1);h^*)-Pr(Y(T)<(-k) \wedge (-m_1), M_Y(0,T) >-m_1; h^*), \label{ordinary2}
\end{align}
where $Y(t)=-X(t)$ and $M_Y(0,T)$ is the $\max_{0 \leq t \leq T} Y(t)$. Then using (\ref{reflection2}) and (\ref{ordinary2}), Equation (\ref{ordinary1}) turns to
\begin{align*}
&S(0)\Phi\left(\frac{(-k) \wedge (-m_1)+(r+\sigma^2/2)T}{\sigma \sqrt{T}}\right)-S(0)e^{\left(\frac{2r}{\sigma^2}+1\right)m_1}\Phi\left( \frac{(-k) \wedge (-m_1)+2m_1+(r+\sigma^2/2)T}{\sigma \sqrt{T}} \right)\\
&-Ke^{-rT} \Phi\left(\frac{(-k) \wedge (-m_1)+(r-\sigma^2/2)T}{\sigma \sqrt{T}}\right)+Ke^{-rT}e^{\left(\frac{2r}{\sigma^2}-1\right)m_1}\Phi\left( \frac{(-k) \wedge (-m_1)+2m_1+(r-\sigma^2/2)T}{\sigma \sqrt{T}} \right).
\end{align*}
It is easy to check that the above formula is the same as the widely known formula; for example, one given in \cite{haug2007}. 
 
\end{example} 

\begin{table}[hbt!]
\small
\centering
\caption{Pricing formulas for icicled step barrier options: Up-barrier cases}
\label{uppricetable}
\begin{tabular}{|c|c|}
\hiderowcolors
\hline
 Option & Price\\
 \hline
 UOP & $Ke^{-rT}PA_u(r-\frac{\sigma^2}{2},  x_1, \cdots, x_n \wedge k, m_{i_1}, \cdots, m_{i_g})$\\
 	&$-S(0)PA_u(r+\frac{\sigma^2}{2}, x_1, \cdots, x_n \wedge k, m_{i_1}, \cdots, m_{i_g})$\\
 \hline
 UIP & $Ke^{-rT}[\Phi(\frac{k-(r-\frac{\sigma^2}{2})t_n}{\sigma \sqrt{t_n}})-PA_u(r-\frac{\sigma^2}{2}, x_1, \cdots, x_n \wedge k, m_{i_1}, \cdots, m_{i_g})]$\\
 &$-S(0)[\Phi(\frac{k-(r+\frac{\sigma^2}{2})t_n}{\sigma \sqrt{t_n}})-PA_u(r+\frac{\sigma^2}{2}, x_1, \cdots, x_n \wedge k, m_{i_1}, \cdots, m_{i_g})]$\\
 \hline
 UOC &$S(0)[PA_u(r+\frac{\sigma^2}{2}, x_1, \cdots, x_{n-1}, x_n, m_{i_1}, \cdots, m_{i_g})$\\
         &$-PA_u(r+\frac{\sigma^2}{2}, x_1, \cdots, x_{n-1}, k, m_{i_1}, \cdots, m_{i_g})]$\\
 ($K<L_n)$& $ -Ke^{-rT}[PA_u(r-\frac{\sigma^2}{2}, x_1, \cdots, x_{n-1}, x_n, m_{i_1}, \cdots, m_{i_g})$\\
                &$-PA_u(r-\frac{\sigma^2}{2}, x_1, \cdots, x_{n-1}, k, m_{i_1}, \cdots, m_{i_g})]$\\
 \hline
 UIC & $S(0)[\Phi(-\frac{k-(r+\frac{\sigma^2}{2})t_n}{\sigma \sqrt{t_n}})-PA_u( r+\frac{\sigma^2}{2}, x_1, \cdots, x_{n-1}, x_n, m_{i_1}, \cdots, m_{i_g})$\\
 ($K<L_n$) &$+ PA_u(r+\frac{\sigma^2}{2}, x_1, \cdots, x_{n-1}, k, m_{i_1}, \cdots, m_{i_g}) ]$\\
 		&$-Ke^{-rT}[\Phi(-\frac{k-(r-\frac{\sigma^2}{2})t_n}{\sigma \sqrt{t_n}})-PA_u(r-\frac{\sigma^2}{2}, x_1, \cdots, x_{n-1}, x_n, m_{i_1}, \cdots, m_{i_g})$\\
 &$+PA_u(r-\frac{\sigma^2}{2}, x_1, \cdots, x_{n-1}, k, m_{i_1}, \cdots, m_{i_g})]$\\
 \hline
 UIC & \multirow{ 2}{*}{$S(0)\Phi(-\frac{k-(r+\frac{\sigma^2}{2})t_n}{\sigma \sqrt{t_n}})-Ke^{-rT}\Phi(-\frac{k-(r-\frac{\sigma^2}{2})t_n}{\sigma \sqrt{t_n}})$}\\
 ($K>L_n$) & \\
 \hline
\end{tabular}
\end{table}
\bigskip
\begin{table}[hbt!]
\small
\centering
\caption{Pricing formulas for icicled step barrier options: Down-barrier cases}
\label{downpricetable}
\begin{tabular}{|c|c|}
\hiderowcolors
\hline
 Option & Price\\
 \hline
 DOC & $S(0)PA_u(-r-\frac{\sigma^2}{2}, -x_1, \cdots, (-x_n) \wedge (-k), -m_{i_1}, \cdots, -m_{i_g})$\\
                             & $-Ke^{-rT}PA_u(-r+\frac{\sigma^2}{2},  -x_1, \cdots, (-x_n) \wedge (-k), -m_{i_1}, \cdots, -m_{i_g})$\\
 \hline
 DIC & $S(0)[\Phi(\frac{-k+(r+\frac{\sigma^2}{2})t_n}{\sigma \sqrt{t_n}})-PA_u(-r-\frac{\sigma^2}{2}, -x_1, \cdots, (-x_n) \wedge (-k), -m_{i_1}, \cdots, -m_{i_g})]$\\
     &$-Ke^{-rT}[\Phi(\frac{-k+(r-\frac{\sigma^2}{2})t_n}{\sigma \sqrt{t_n}})-PA_u(-r+\frac{\sigma^2}{2}, -x_1, \cdots, (-x_n) \wedge (-k), -m_{i_1}, \cdots, -m_{i_g})]$\\
  \hline
 DOP & $ Ke^{-rT}[PA_u(-r+\frac{\sigma^2}{2}, -x_1, \cdots, -x_{n-1}, -x_n, -m_{i_1}, \cdots, -m_{i_g})$\\
 				&$-PA_u(-r+\frac{\sigma^2}{2}, -x_1, \cdots, -x_{n-1}, -k, -m_{i_1}, \cdots, -m_{i_g})]$\\
     ($K>L_n)$   &$-S(0)[PA_u(-r-\frac{\sigma^2}{2}, -x_1, \cdots, -x_{n-1}, -x_n, -m_{i_1}, \cdots, -m_{i_g})$\\
                       &$-PA_u(-r-\frac{\sigma^2}{2}, -x_1, \cdots, -x_{n-1}, -k, -m_{i_1}, \cdots, -m_{i_g})]$\\
  \hline
 DIP & $Ke^{-rT}[\Phi(\frac{k-(r-\frac{\sigma^2}{2})t_n}{\sigma \sqrt{t_n}})-PA_u(-r+\frac{\sigma^2}{2}, -x_1, \cdots, -x_{n-1}, -x_n, -m_{i_1}, \cdots, -m_{i_g})$\\
               &$+PA_u(-r+\frac{\sigma^2}{2}, -x_1, \cdots, -x_{n-1}, -k, -m_{i_1}, \cdots, -m_{i_g})]-S(0)[\Phi(\frac{k-(r+\frac{\sigma^2}{2})t_n}{\sigma \sqrt{t_n}})$\\
   ($K>L_n$)&$ -PA_u( -r-\frac{\sigma^2}{2}, -x_1, \cdots, -x_{n-1}, -x_n, -m_{i_1}, \cdots, -m_{i_g})$\\
           &$+ PA_u(-r-\frac{\sigma^2}{2}, -x_1, \cdots, -x_{n-1}, -k, -m_{i_1}, \cdots, -m_{i_g}) ]$\\
 \hline
 DIP & \multirow{ 2}{*}{$Ke^{-rT}\Phi(\frac{k-(r-\frac{\sigma^2}{2})t_n}{\sigma \sqrt{t_n}})-S(0)\Phi(\frac{k-(r+\frac{\sigma^2}{2})t_n}{\sigma \sqrt{t_n}})$}\\
 ($K<L_n$) & \\
 \hline
\end{tabular}
\end{table}

Table \ref{upoptions} presents the multi-step up-barrier option prices for various input parameters such as the interest rate, the strike price, and the volatility of the underlying asset. The option maturity is assumed to be a half years ($T=0.5$), and the barrier levels are changeable at the end of every month($t_1=1/12, t_2=2/12, \cdots, t_6=T=6/12$). We assume $S(0)=100$ and leave icicles out of the consideration for a simpler illustration. Three types of barrier levels are considered; Type 1 has an increasing pattern, Type 2 has a decreasing pattern, and Type 3 has an increasing and then decreasing pattern as shown in Figure \ref{upbarriertype}. Columns of UOC, UIC, UOP, and UIP provide the corresponding multi-step barrier option prices; for instance, UOC is the up and out multi-step call option price.

The overall sensitivity behaviors of UIC, UOP, and UIP option prices with respect to the interest rate, the volatility, or the strike price do not differ from those of standard call and put option prices. 
UOC option prices also behave similarly as other options with respect to the strike price but they do not increase as the interest rate increases, although the price difference is very small. Also, with respect to the volatility, they move in the opposite direction compared to the standard call option. When $\sigma$ gets larger, option prices generally get larger since the underlying price movement becomes wider, but in case of UOC, a larger value of $\sigma$ will increase the probability of being knocked out which seems to have a bigger impact on the price. The same reasoning can be applied to explain the phenomenon that UOP prices with barrier Types 1 and 3 do not increase as $\sigma$ increases when $K=110$.
Note also that UOC option prices in Table \ref{upoptions} are all zero for Type 3 barrier with $K=110$. The up-barrier level with Type 3 is equal to the strike price at the last subinterval, so the payoff of the UOC option will be automatically 0 which makes the price equal to 0. 
 
Table \ref{k100table_BB_RMS} compares the exact option prices calculated from the formula obtained in this section with the results of Monte Carlo(MC) simulation using Brownian bridge technique. Since Monte Carlo simulation checks if the underlying asset price hits the barriers only at discrete time points, the simulated price from basic MC simulation methods is systematically biased. This bias is sometimes called as monitoring bias, which is of order $O(M^{-1/2})$ where $M$ is the number of discretization points(See \cite{gobet2009}). \cite{korn2010} suggests filling the gap between discretization points using a Brownian bridge instead of the linear interpolation in order to remove the monitoring bias. We followed the Brownian bridge technique given in \cite{korn2010} for our setting. In each cell of Table \ref{k100table_BB_RMS}, four values are shown; the first value is the multi-step barrier option price calculated from the formula, the second value in the curly brackets is the MC estimated option price, the third value in parentheses is the standard error (SE) of the repeated MC experiments, and the fourth value in squared brackets is the absolute error, which is the model price using multi-step probabilities subtracted from the MC estimated price. In general, SEs and absolute errors of MC estimates are very small, although there are some random simulation errors left. The MC simulation used 1,000,000 sample paths. We also calculate the root mean square (RMS) relative error which is defined as
	$$ \mbox{RMS relative error }=\sqrt{\frac{1}{m} \sum_{i=1}^m (C_i-\tilde{C}_i)^2/C_i^2},$$
by \cite{glasserman2004}. Here, $m$ is the number of cells in each table, $C_i$ is the price by the explicit option pricing formula for the $i$th cell, and $\tilde{C}_i$ is the estimated price using Monte Carlo simulation for the $i$th cell.  The RMS relative error of the Monte Carlo simulation for up-barrier options in Table \ref{k100table_BB_RMS} is 0.0035(0.35\%). A usage of RMS relative error can be also found in \cite{abg1998}.


The overall behaviors of down-barrier options in Table \ref{downoptions} are very similar to up-barrier options. In general, the movement patterns of DOC, DIC, and DIP option prices with respect to the interest rate, the volatility, or the strike price do not differ from those of standard option prices. But DOP option prices move differently to some extent. DOP option prices  with respect to the volatility move in the opposite direction compared to the standard call option. As with UOC cases, a larger value of $\sigma$ will increase the probability of being knocked out which seems to have a bigger impact on the price so that the option prices decrease as $\sigma$ increases. The same reasoning can be applied to explain the decrease in DOC prices with Type 2 and Type 3 barriers when $K=90$. DOP behaves similarly as other options with respect to the strike price.

Note that DOP prices with Type 3 barrier when $K=90$ are all zero, since the option must be killed before the expiration for the option payoff, $(90-S(T))^+$ to be greater than 0. Another thing to note is that Type 3 barrier is always equal to or higher than Type 2 barrier in this setting. Therefore, down-and-in options with Type 3 barrier are more expensive than those with Type 2 barrier and down-and-out options with Type 2 are more expensive than those with Type 3. 

Table \ref{k100dtable_BB} compares the exact option prices calculated from the formula obtained in this section with the results of Monte Carlo(MC) simulation using Brownian bridge technique given in \cite{korn2010}. In each cell of 
Table \ref{k100dtable_BB}, four values are shown; the first value is the multi-step barrier option price calculated from the formula, the second value in the curly brackets is the MC estimated option price, the third value in parentheses is the standard error (SE) of the repeated MC experiments, and the fourth value in squared brackets is the absolute error, which is the model price using multi-step probabilities subtracted from the MC estimated price. Again, SEs and absolute errors of MC estimates are very small in general and the RMS relative error of the Monte Carlo simulation for down-barrier options in Table \ref{k100dtable_BB}  is 0.0052(0.52\%). The MC simulation used 1,000,000 sample paths. 


\begin{table}
\small
\caption{Multi-step up-barrier option prices}
\label{uptable}

\begin{subtable}{0.9\textwidth}
\caption{Multi-step up-barrier option prices at various input levels with $S(0)=100$, $t_1=1/12$, $t_2=2/12$, $\cdots$, $t_n=T=6/12$; Barrier types of 1, 2, and 3 are displayed in Figure \ref{upbarriertype}. }
\label{upoptions}
\renewcommand{\tabcolsep}{1.5mm}
\renewcommand{\arraystretch}{0.8}
\begin{tabular}{| c | c | c | c c c c| c c  c c| c  c c c |} \hiderowcolors\cline{3-15}
                   \multicolumn{2}{c|}{} &Barrier     & \multicolumn{4}{c|}{$K=90$}     & \multicolumn{4}{c|}{$K=100$}   & \multicolumn{4}{c|}{$K=110$}    \\ 
                   \cline{1-2} \cline{4-15}                                  
$r$                     & $\sigma$                & type & UOC  & UIC   & UOP  & UIP  & UOC  & UIC  & UOP  & UIP  & UOC  & UIC  & UOP   & UIP  \\ \hline
\multirow{6}{*}{0.03} & \multirow{3}{*}{0.2} & 1    & 2.58 & 10.22 & 1.16 & 0.30 & 0.67 & 5.71 & 3.54 & 1.34 & 0.05 & 2.56 & 7.23  & 3.75 \\
                      &                      & 2    & 3.81 & 8.99  & 1.46 & 0.002 & 0.75 & 5.62 & 4.85 & 0.04 & 0.004 & 2.61 & 10.55 & 0.42 \\
                      &                      & 3    & 1.91 & 10.89 & 1.18 & 0.28 & 0.28 & 6.09 & 3.65 & 1.23 & 0 & 2.61 & 7.47  & 3.51 \\ \cline{2-15}
                      & \multirow{3}{*}{0.3} & 1    & 0.97 & 13.91 & 2.10 & 1.44 & 0.24 & 8.90 & 4.12 & 3.54 & 0.02 & 5.22 & 6.63  & 6.97 \\
                      &                      & 2    & 1.77 & 13.11 & 3.46 & 0.08 & 0.30 & 8.84 & 7.28 & 0.38 & 0.001 & 5.24 & 12.25 & 1.35 \\
                      &                      & 3    & 0.67 & 14.21 & 2.17 & 1.38 & 0.09 & 9.06 & 4.29 & 3.37 & 0 & 5.24 & 6.91  & 6.69 \\ \cline{1-15}
\multirow{6}{*}{0.04} & \multirow{3}{*}{0.2} & 1    & 2.59 & 10.56 & 1.08 & 0.28 & 0.67 & 5.95 & 3.36 & 1.28 & 0.05 & 2.70 & 6.93  & 3.64 \\
                      &                      & 2    & 3.80 & 9.35  & 1.36 & 0.001 & 0.75 & 5.87 & 4.61 & 0.04 & 0.004 & 2.75 & 10.16 & 0.42 \\
                      &                      & 3    & 1.90 & 11.25 & 1.10 & 0.26 & 0.28 & 6.34 & 3.47 & 1.18 & 0 & 2.76 & 7.16  & 3.41 \\ \cline{2-15}
                      & \multirow{3}{*}{0.3} & 1    & 0.97 & 14.22 & 2.01 & 1.39 & 0.25 & 9.14 & 3.97 & 3.44 & 0.02 & 5.39 & 6.42  & 6.81 \\
                      &                      & 2    & 1.76 & 13.42 & 3.32 & 0.08 & 0.31 & 9.09 & 7.04 & 0.37 & 0.001 & 5.41 & 11.90 & 1.33 \\
                      &                      & 3    & 0.67 & 14.52 & 2.08 & 1.32 & 0.09 & 9.30 & 4.14 & 3.27 & 0 & 5.41 & 6.69  & 6.54 \\ \cline{1-15}
\end{tabular}
\end{subtable}

\bigskip
\begin{subtable}{0.9\textwidth}
\caption{Monte Carlo simulated up-barrier option prices; in each cell, the first number is the multi-step option price, the number in the curly brackets is the MC simulated price, the number in the parentheses is the standard error of the MC experiments, and the number in the square brackets is the difference between the simulated price and the multi-step formula price. The simulation used 1,000,000 sample paths. RMS relative error of this simulation is 0.0035.}
\label{k100table_BB_RMS}
\renewcommand{\tabcolsep}{1.5mm}
\renewcommand{\arraystretch}{0.8}
\begin{tabular}{| c | c | c | c | c | c |c|c|c|c|} \hiderowcolors\cline{2-10}
\multicolumn{1}{c|}{}&Barrier&\multicolumn{4}{c|}{$r=0.03$}& \multicolumn{4}{c|}{$r=0.04$}\\ 
\cline{1-1} \cline{3-10} 
{$\sigma$} &{Type} & {UOC} & {UIC} & {UOP} & {UIP}& {UOC} & {UIC} & {UOP} & {UIP}\\ \hline

\multirow{12}{0.5cm}{0.2}
&\multirow{4}{0.5cm}{1}
&0.6654&5.7057&3.5436&1.3386&0.6746&5.9525&3.3635&1.2834\\&&\{0.6642\}&\{5.7097\}&\{3.5353\}&\{1.3385\}&\{0.6733\}&\{5.9558\}&\{3.3670\}&\{1.2859\}\\&&(0.0018)&(0.0092)&(0.0062)&(0.0030)&(0.0019)&(0.0098)&(0.0059)&(0.0030)\\&&[-0.0011]&[0.0041]&[-0.0083]&[-0.0001]&[-0.0013]&[0.0033]&[0.0035]&[0.0024]\\ \cline{2-10}

&\multirow{4}{0.5cm}{2}
&0.7467&5.6243&4.8459&0.0363&0.7527&5.8744&4.6115&0.0354\\&&\{0.7484\}&\{5.6142\}&\{4.8425\}&\{0.0361\}&\{0.7530\}&\{5.8537\}&\{4.6163\}&\{0.0358\}\\&&(0.0017)&(0.0097)&(0.0071)&(0.0004)&(0.0017)&(0.0098)&(0.0069)&(0.0004)\\&&[0.0017]&[-0.0101]&[-0.0034]&[-0.0002]&[0.0003]&[-0.0207]&[0.0048]&[0.0004]\\ \cline{2-10}

&\multirow{4}{0.5cm}{3}
&0.2831&6.0880&3.6499&1.2323&0.2846&6.3425&3.4660&1.1809\\&&\{0.2837\}&\{6.0792\}&\{3.6541\}&\{1.2320\}&\{0.2848\}&\{6.3330\}&\{3.4744\}&\{1.1783\}\\&&(0.0009)&(0.0095)&(0.0061)&(0.0028)&(0.0009)&(0.0102)&(0.0059)&(0.0027)\\&&[0.0006]&[-0.0088]&[0.0041]&[-0.0003]&[0.0002]&[-0.0095]&[0.0084]&[-0.0026]\\ \cline{1-10}

\multirow{12}{0.5cm}{0.3}
&\multirow{4}{0.5cm}{1}
&0.2450&8.9044&4.1165&3.5441&0.2461&9.1443&3.9695&3.4408\\&&\{0.2434\}&\{8.9132\}&\{4.1108\}&\{3.5391\}&\{0.2450\}&\{9.1240\}&\{3.9702\}&\{3.4536\}\\&&(0.0010)&(0.0145)&(0.0079)&(0.0060)&(0.0011)&(0.0144)&(0.0080)&(0.0059)\\&&[-0.0015]&[0.0087]&[-0.0057]&[-0.0050]&[-0.0012]&[-0.0203]&[0.0007]&[0.0128]\\ \cline{2-10}

&\multirow{4}{0.5cm}{2}
&0.3045&8.8449&7.2786&0.3820&0.3051&9.0854&7.0356&0.3747\\&&\{0.3046\}&\{8.8474\}&\{7.2818\}&\{0.3819\}&\{0.3051\}&\{9.1019\}&\{7.0268\}&\{0.3711\}\\&&(0.0010)&(0.0149)&(0.0101)&(0.0018)&(0.0010)&(0.0147)&(0.0099)&(0.0017)\\&&[0.0001]&[0.0024]&[0.0032]&[-0.0001]&[0.0000]&[0.0166]&[-0.0089]&[-0.0036]\\ \cline{2-10}

&\multirow{4}{0.5cm}{3}
&0.0872&9.0622&4.2895&3.3711&0.0873&9.3032&4.1380&3.2723\\&&\{0.0869\}&\{9.0894\}&\{4.2880\}&\{3.3692\}&\{0.0884\}&\{9.3034\}&\{4.1430\}&\{3.2753\}\\&&(0.0004)&(0.0148)&(0.0077)&(0.0057)&(0.0004)&(0.0145)&(0.0078)&(0.0056)\\&&[-0.0002]&[0.0272]&[-0.0015]&[-0.0019]&[0.0011]&[0.0002]&[0.0049]&[0.0030]\\

 \hline

\end{tabular}
\end{subtable}

\end{table}


 \begin{figure}[H]
\captionsetup[subfigure]{justification=centering}
\centering
\begin{subfigure}[b]{0.85\textwidth}
\centering
  \includegraphics[width=0.7\textwidth]{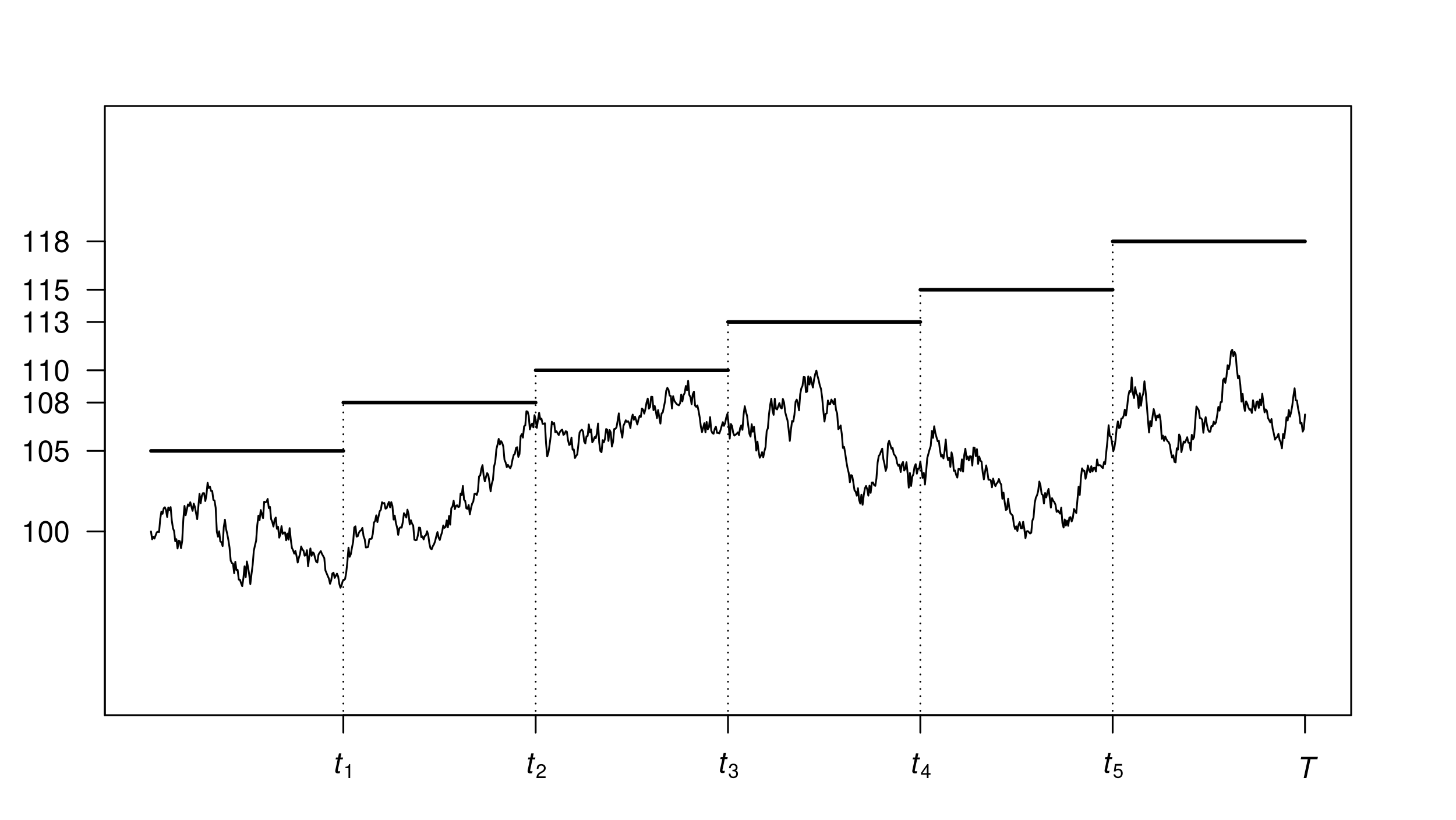}
\caption{Type1 up barrier }
   \end{subfigure}
\begin{subfigure}[b]{0.85\textwidth}
 \centering
 \includegraphics[width=0.7\textwidth]{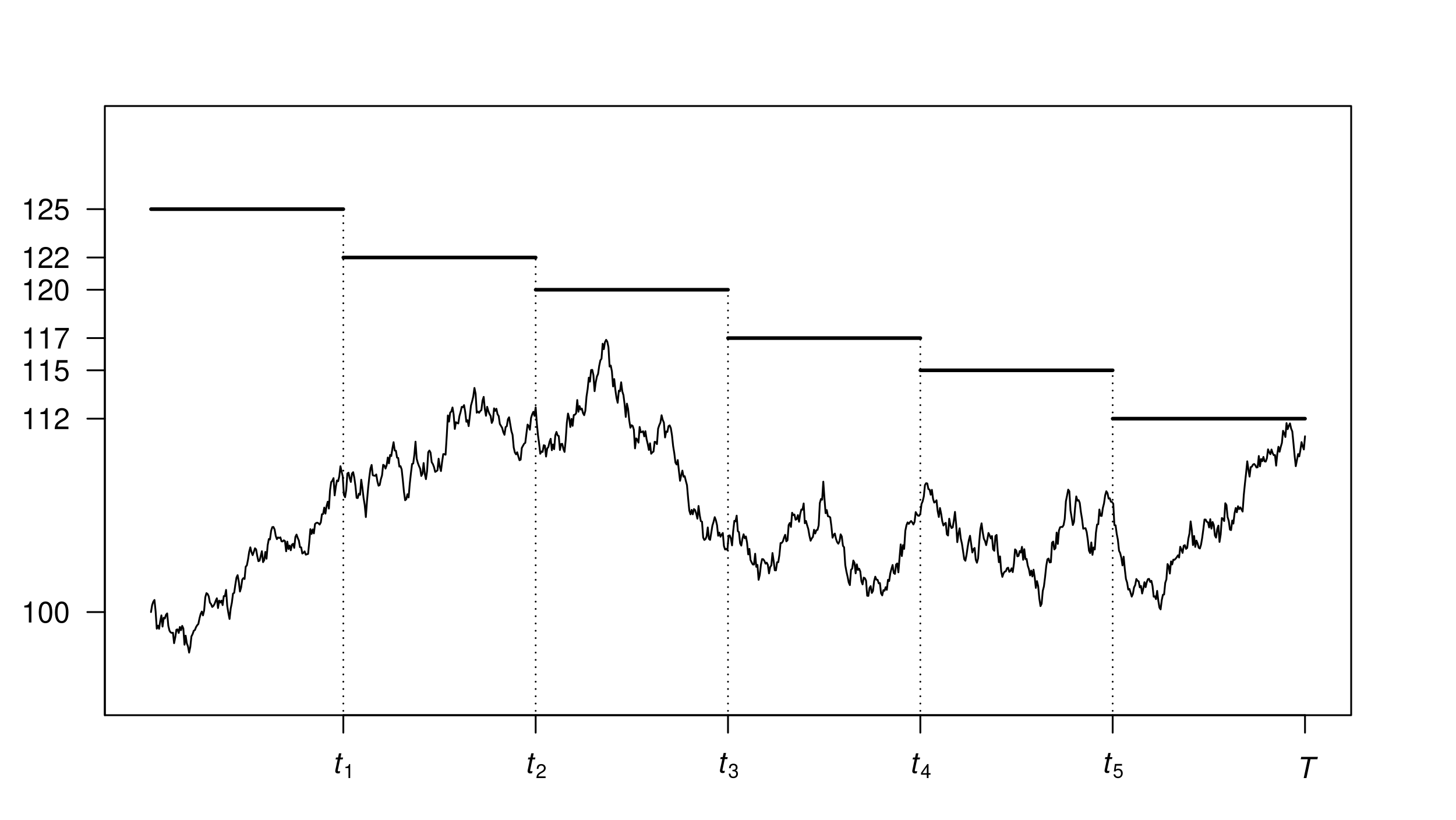}
\caption{Type2 up barrier }
   \end{subfigure}

\begin{subfigure}[b]{0.85\textwidth}
  \centering
  \includegraphics[width=0.7\textwidth]{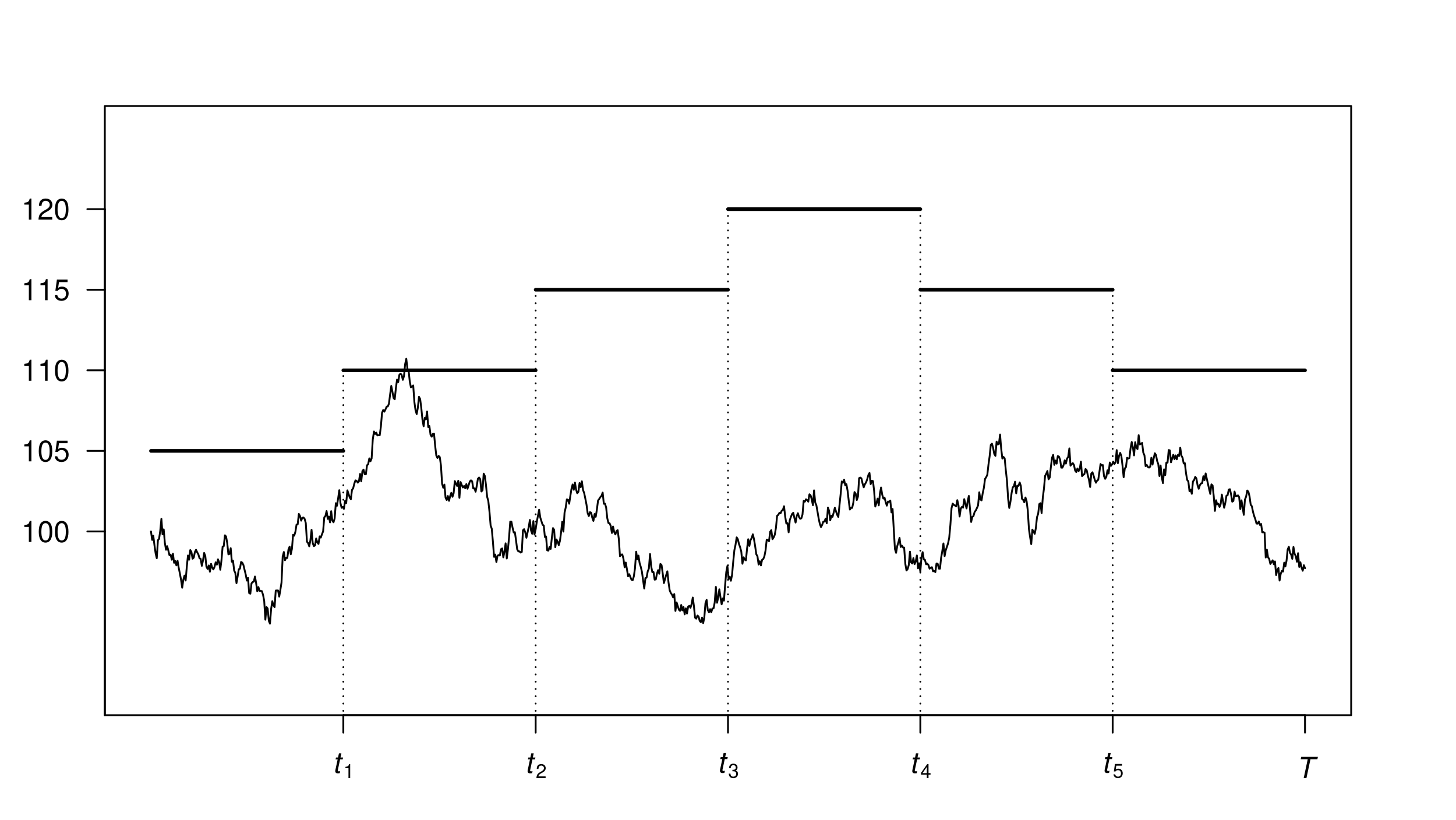}
\caption{Type3 up barrier }
   \end{subfigure}
\caption {Three types of multi-step up barriers used in Table \ref{uptable}}
\label{upbarriertype}
\end{figure}


\begin{table}
\small
\caption{Multi-step down-barrier option prices}
\label{downtable}
\begin{subtable}{0.9\textwidth}
\caption{Multi-step down-barrier option prices at various input levels with $S(0)=100$, $t_1=1/12$, $t_2=2/12$, $\cdots$, $t_n=T=6/12$; Barrier types of 1, 2, and 3 are displayed in Figure \ref{downbarriertype}.}
\label{downoptions}
\renewcommand{\tabcolsep}{1.5mm}
\renewcommand{\arraystretch}{0.8}
\begin{tabular}{| c | c | c | c c c c| c c  c c| c  c c c |} \hiderowcolors\cline{3-15}
                   \multicolumn{2}{c|}{} &Barrier     & \multicolumn{4}{c|}{$K=90$}     & \multicolumn{4}{c|}{$K=100$}   & \multicolumn{4}{c|}{$K=110$}    \\ 
                   \cline{1-2} \cline{4-15}                                  
$r$                     & $\sigma$                & type & DOC   & DIC  & DOP  & DIP  & DOC  & DIC  & DOP  & DIP  & DOC  & DIC  & DOP  & DIP   \\ \hline
\multirow{6}{*}{0.03} & \multirow{3}{*}{0.2} & 1    & 12.57 & 0.23 & 0.006 & 1.45 & 6.36 & 0.01 & 0.97 & 3.91 & 2.61 & 0.001 & 4.39 & 6.58  \\
                      &                      & 2    & 9.02  & 3.77 & 0.056 & 1.40 & 4.87 & 1.50 & 0.80 & 4.08 & 2.14 & 0.47 & 2.96 & 8.01  \\
                      &                      & 3    & 8.74  & 4.06 & 0 & 1.46 & 4.86 & 1.51 & 0.17 & 4.71 & 2.14 & 0.47 & 1.51 & 9.46  \\ \cline{2-15}
                      & \multirow{3}{*}{0.3} & 1    & 14.02 & 0.86 & 0.003 & 3.54 & 8.95 & 0.20 & 0.45 & 7.21 & 5.19 & 0.04 & 2.23 & 11.37 \\
                      &                      & 2    & 7.81  & 7.07 & 0.026 & 3.52 & 5.26 & 3.89 & 0.33 & 7.33 & 3.26 & 1.98 & 1.19 & 12.41 \\
                      &                      & 3    & 7.58  & 7.30 & 0 & 3.54 & 5.21 & 3.94 & 0.05 & 7.61 & 3.25 & 1.99 & 0.51 & 13.10 \\ \cline{1-15}
\multirow{6}{*}{0.04} & \multirow{3}{*}{0.2} & 1    & 12.92 & 0.23 & 0.005 & 1.36 & 6.61 & 0.01 & 0.95 & 3.70 & 2.76 & 0.001 & 4.33 & 6.25  \\
                      &                      & 2    & 9.31  & 3.83 & 0.053 & 1.31 & 5.08 & 1.54 & 0.78 & 3.87 & 2.26 & 0.50 & 2.91 & 7.67  \\
                      &                      & 3    & 9.03  & 4.11 & 0 & 1.37 & 5.06 & 1.56 & 0.17 & 4.48 & 2.26 & 0.50 & 1.50 & 9.08  \\ \cline{2-15}
                      & \multirow{3}{*}{0.3} & 1    & 14.33 & 0.86 & 0.002 & 3.40 & 9.19 & 0.20 & 0.45 & 6.96 & 5.37 & 0.05 & 2.21 & 11.02 \\
                      &                      & 2    & 8.01  & 7.18 & 0.025 & 3.38 & 5.41 & 3.98 & 0.33 & 7.08 & 3.37 & 2.04 & 1.18 & 12.05 \\
                      &                      & 3    & 7.77  & 7.41 & 0 & 3.40 & 5.37 & 4.02 & 0.05 & 7.36 & 3.36 & 2.05 & 0.50 & 12.73 \\ \cline{1-15}
\end{tabular}
\end{subtable}

\bigskip
\begin{subtable}{0.9\textwidth}
\caption{ Monte Carlo simulated down-barrier option prices; in each cell, the first number is the multi-step option price, the number in the curly brackets is the MC simulated price, the number in the parentheses is the standard error of the MC experiments, and the number in the square brackets is the difference between the simulated price and the multi-step formula price. The simulation used 1,000,000 sample paths. RMS relative error of this simulation is 0.0052.}
\label{k100dtable_BB}
\renewcommand{\tabcolsep}{1.5mm}
\renewcommand{\arraystretch}{0.8}
\begin{tabular}{| c | c | c | c | c | c |c|c|c|c|} \hiderowcolors\cline{2-10}
\multicolumn{1}{c|}{}&Barrier&\multicolumn{4}{c|}{$r=0.03$}& \multicolumn{4}{c|}{$r=0.04$}\\ 
\cline{1-1} \cline{3-10} 
{$\sigma$} &Type & {DOC} & {DIC} & {DOP} & {DIP}  & {DOC} & {DIC} & {DOP} & {DIP}\\ \hline

\multirow{12}{0.5cm}{0.2}
&\multirow{4}{0.5cm}{1}
&6.3590&0.0120&0.9676&3.9146&6.6149&0.0121&0.9451&3.7018\\&&\{6.3620\}&\{0.0120\}&\{0.9661\}&\{3.9076\}&\{6.6166\}&\{0.0125\}&\{0.9480\}&\{3.7049\}\\&&(0.0091)&(0.0002)&(0.0019)&(0.0072)&(0.0097)&(0.0002)&(0.0019)&(0.0067)\\&&[0.0030]&[-0.0001]&[-0.0014]&[-0.0070]&[0.0017]&[0.0004]&[0.0029]&[0.0030]\\ \cline{2-10}

&\multirow{4}{0.5cm}{2}
&4.8744&1.4966&0.8024&4.0798&5.0833&1.5438&0.7793&3.8677\\&&\{4.8695\}&\{1.4931\}&\{0.8006\}&\{4.0779\}&\{5.0711\}&\{1.5356\}&\{0.7805\}&\{3.8716\}\\&&(0.0087)&(0.0035)&(0.0020)&(0.0069)&(0.0085)&(0.0036)&(0.0020)&(0.0069)\\&&[-0.0049]&[-0.0035]&[-0.0018]&[-0.0019]&[-0.0122]&[-0.0082]&[0.0013]&[0.0039]\\ \cline{2-10}

&\multirow{4}{0.5cm}{3}
&4.8563&1.5148&0.1728&4.7094&5.0649&1.5621&0.1698&4.4771\\&&\{4.8541\}&\{1.5088\}&\{0.1718\}&\{4.7142\}&\{5.0576\}&\{1.5602\}&\{0.1693\}&\{4.4834\}\\&&(0.0085)&(0.0035)&(0.0006)&(0.0070)&(0.0092)&(0.0036)&(0.0006)&(0.0066)\\&&[-0.0021]&[-0.0060]&[-0.0010]&[0.0048]&[-0.0074]&[-0.0019]&[-0.0005]&[0.0063]\\ \cline{1-10}

\multirow{12}{0.5cm}{0.3}
&\multirow{4}{0.5cm}{1}
&8.9480&0.2014&0.4537&7.2069&9.1872&0.2033&0.4484&6.9619\\&&\{8.9577\}&\{0.1989\}&\{0.4555\}&\{7.1944\}&\{9.1668\}&\{0.2022\}&\{0.4481\}&\{6.9758\}\\&&(0.0144)&(0.0013)&(0.0013)&(0.0108)&(0.0145)&(0.0013)&(0.0013)&(0.0108)\\&&[0.0097]&[-0.0025]&[0.0018]&[-0.0125]&[-0.0204]&[-0.0011]&[-0.0004]&[0.0139]\\ \cline{2-10}

&\multirow{4}{0.5cm}{2}
&5.2580&3.8914&0.3338&7.3268&5.4124&3.9780&0.3289&7.0814\\&&\{5.2602\}&\{3.8918\}&\{0.3348\}&\{7.3289\}&\{5.4206\}&\{3.9864\}&\{0.3287\}&\{7.0691\}\\&&(0.0113)&(0.0075)&(0.0012)&(0.0102)&(0.0113)&(0.0073)&(0.0012)&(0.0100)\\&&[0.0023]&[0.0003]&[0.0010]&[0.0021]&[0.0081]&[0.0084]&[-0.0002]&[-0.0123]\\ \cline{2-10}

&\multirow{4}{0.5cm}{3}
&5.2130&3.9364&0.0532&7.6074&5.3671&4.0233&0.05270&7.3576\\&&\{5.2240\}&\{3.9523\}&\{0.0528\}&\{7.6044\}&\{5.3596\}&\{4.0322\}&\{0.0526\}&\{7.3656\}\\&&(0.0115)&(0.0075)&(0.0003)&(0.0104)&(0.0109)&(0.0076)&(0.0003)&(0.0104)\\&&[0.0110]&[0.0159]&[-0.0003]&[-0.0031]&[-0.0076]&[0.0089]&[-0.0001]&[0.0080]\\ 

 \hline

\end{tabular}
\end{subtable}

\end{table}


\begin{figure}[H]
\captionsetup[subfigure]{justification=centering}
\centering
\begin{subfigure}[b]{0.85\textwidth}
 \centering
  \includegraphics[width=0.7\textwidth]{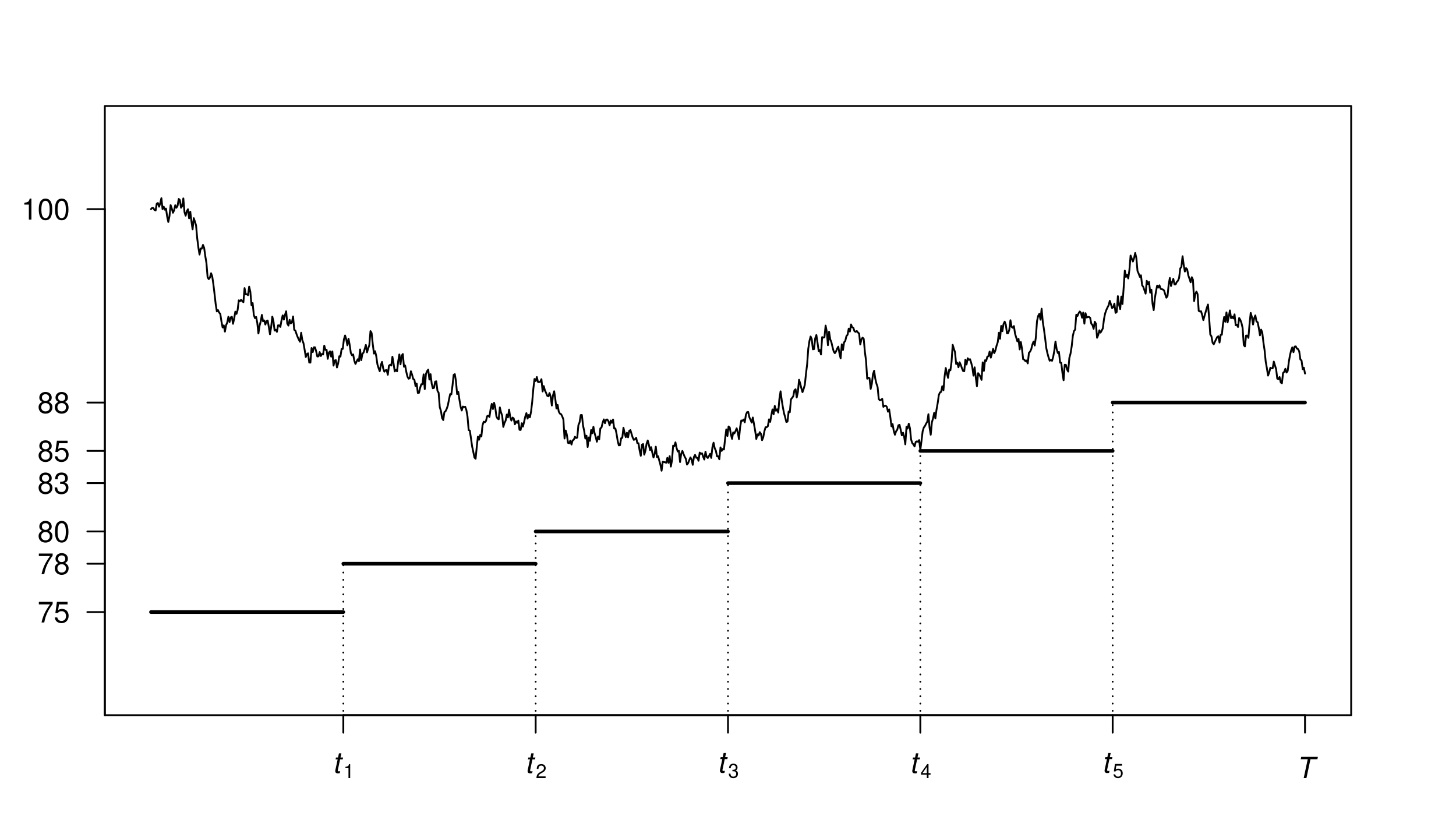}
\caption{Type1 down barrier }
   \end{subfigure}

\begin{subfigure}[b]{0.85\textwidth}
 \centering
 \includegraphics[width=0.7\textwidth]{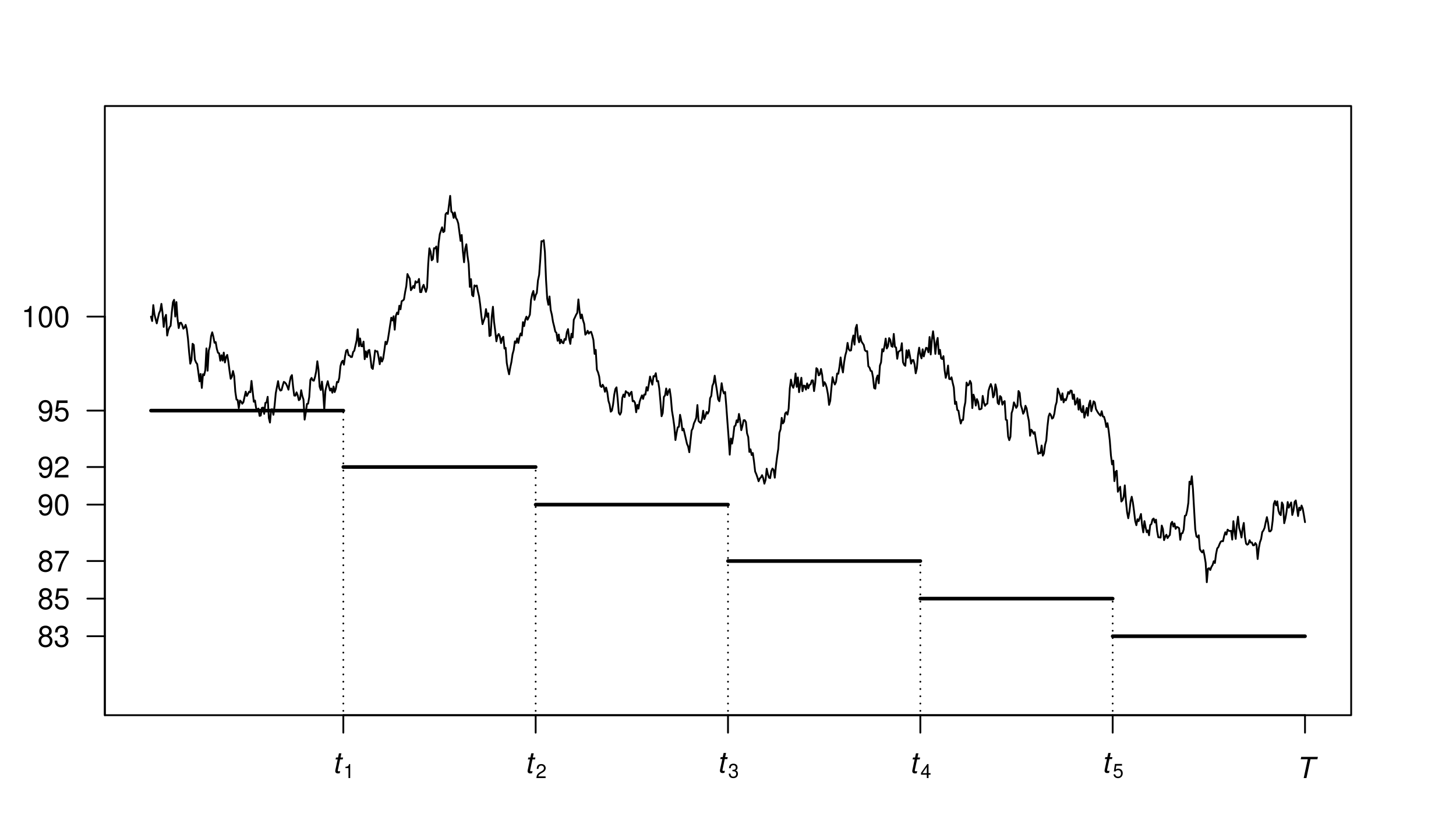}
\caption{Type2 down barrier }
   \end{subfigure}

\begin{subfigure}[b]{0.85\textwidth}
 \centering
  \includegraphics[width=0.7\textwidth]{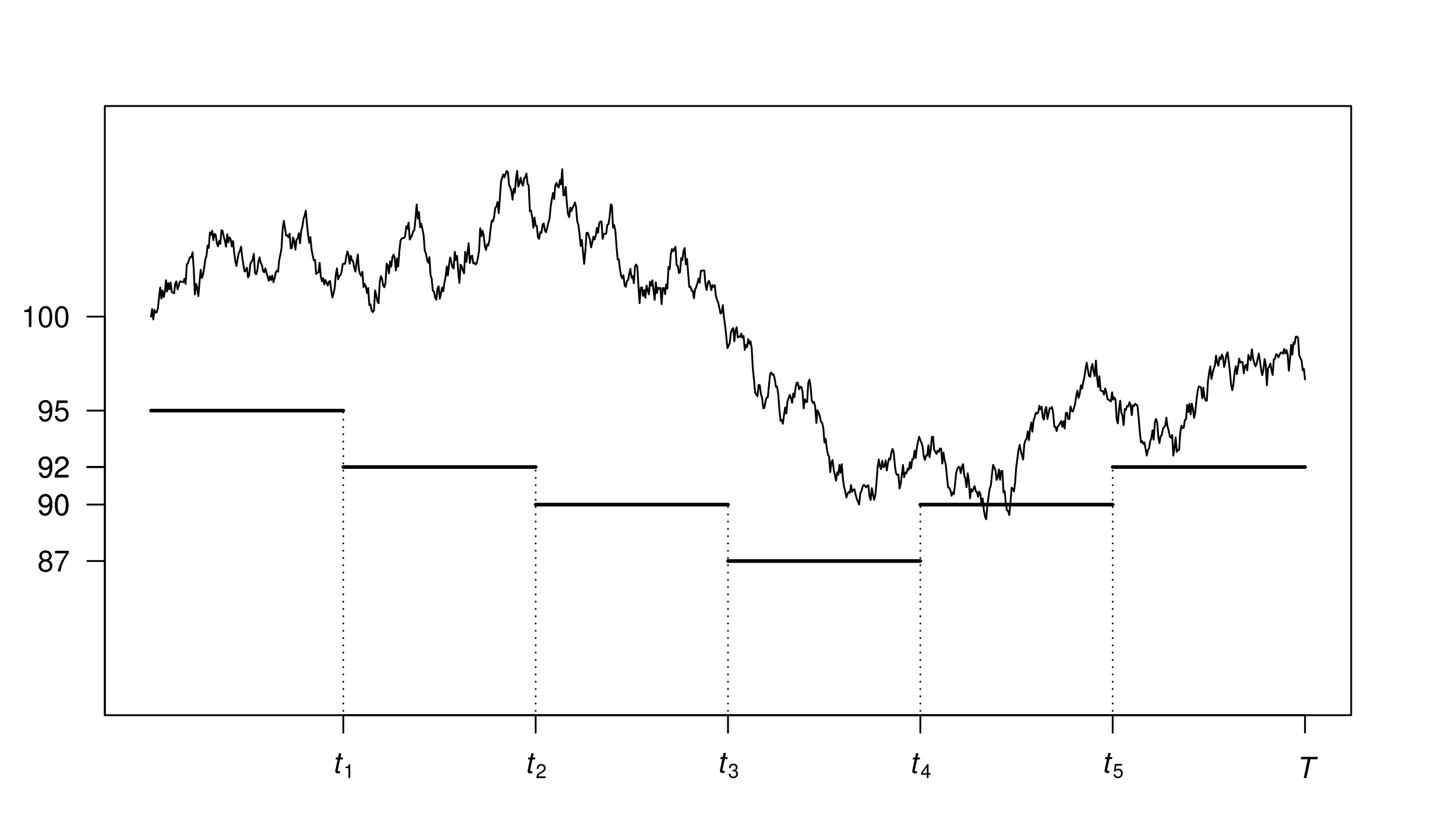}
\caption{Type3 down barrier }
   \end{subfigure}
\caption {Three types of multi-step down barriers used in Table \ref{downtable}}
\label{downbarriertype}
\end{figure}

Some high volatility cases are studied in Table \ref{kdtable_vol}. For the strike price $K=100$, the same interest rates, and the same barrier values as in Tables \ref{k100table_BB_RMS} and \ref{k100dtable_BB}, we calculate up and down barrier option prices for higher volatilities of $\sigma=0.4$ and $\sigma=0.5$. Looking at Table \ref{kdtable_vol} with Tables \ref{k100table_BB_RMS} and \ref{k100dtable_BB} together, we observe the same patterns observed in Tables \ref{k100table_BB_RMS} and \ref{k100dtable_BB} when $\sigma$ increases from 0.2 to 0.5. Columns of ``Call'' and ``Put'' have the standard European call and put option prices from the Black-Scholes formula. As well expected, the sums of UOC and UIC are equal to the standard call option prices and the sums of UOP and UIP are equal to the standard put option prices. The same phenomenon occurs in the case of down-barrier options.

\begin{table}[hbt!] 
\small

\caption{Multi-step option prices with higher volatilities}
\label{kdtable_vol}

\begin{subtable}{0.4\textwidth}
\caption{$K=100$, Up-barrier options}
\label{k100table_vol}
\renewcommand{\tabcolsep}{1.5mm}
\renewcommand{\arraystretch}{0.9}
\begin{tabular}{ | c | c | c c c c c c | c | c | c | c | c | c | } \hiderowcolors\hline

{$r$} & {$\sigma$} &{${B}_{1}$ } &{${B}_{2}$ }&{${B}_{3}$ }&{${B}_{4}$ }&{${B}_{5}$ }&{${B}_{6}$ } & {UOC} & {UIC} & {UOP} & {UIP} &{Call} & {Put}\\ \hline

\multirow{6}{0.5cm}{0.03}

& \multirow{3}{0.3cm}{0.4}  &\multirow{1}{0.3cm}{105}  &\multirow{1}{0.3cm}{108} &\multirow{1}{0.3cm}{110} &\multirow{1}{0.3cm}{113} &\multirow{1}{0.3cm}{115} &\multirow{1}{0.3cm}{118} &0.1124&11.8102&4.3736&6.0602& \multirow{3}*{11.9226} & \multirow{3}*{10.4338} \\
 & &  \multirow{1}{0.3cm}{125}  &\multirow{1}{0.3cm}{122} &\multirow{1}{0.3cm}{120} &\multirow{1}{0.3cm}{117} &\multirow{1}{0.3cm}{115} &\multirow{1}{0.3cm}{112} &0.1467&11.7758&9.2451&1.1886	& & \\
 & & \multirow{1}{0.3cm}{105}  &\multirow{1}{0.3cm}{110} &\multirow{1}{0.3cm}{115} &\multirow{1}{0.3cm}{120} &\multirow{1}{0.3cm}{115} &\multirow{1}{0.3cm}{110} &0.0355&11.8870&4.5632&5.8705	& & \\ \cline{2-14}
 & \multirow{3}{0.3cm}{0.5}  & \multirow{1}{0.3cm}{105}  &\multirow{1}{0.3cm}{108} &\multirow{1}{0.3cm}{110} &\multirow{1}{0.3cm}{113} &\multirow{1}{0.3cm}{115} &\multirow{1}{0.3cm}{118} &0.0598&14.6241&4.5085&8.6866	& \multirow{3}*{14.6840} & \multirow{3}*{13.1952} \\
 & & \multirow{1}{0.3cm}{125}  &\multirow{1}{0.3cm}{122} &\multirow{1}{0.3cm}{120} &\multirow{1}{0.3cm}{117} &\multirow{1}{0.3cm}{115} &\multirow{1}{0.3cm}{112} &0.0804&14.6036&10.8109&2.3842	& & \\
 & & \multirow{1}{0.3cm}{105}  &\multirow{1}{0.3cm}{110} &\multirow{1}{0.3cm}{115} &\multirow{1}{0.3cm}{120} &\multirow{1}{0.3cm}{115} &\multirow{1}{0.3cm}{110} &0.0174&14.6666&4.6943&8.5009	& & \\ \hline

\multirow{6}{0.5cm}{0.04}

& \multirow{3}{0.3cm}{0.4}  & \multirow{1}{0.3cm}{105}  &\multirow{1}{0.3cm}{108} &\multirow{1}{0.3cm}{110} &\multirow{1}{0.3cm}{113} &\multirow{1}{0.3cm}{115} &\multirow{1}{0.3cm}{118} &0.1126&12.0401&4.2513&5.9213 & \multirow{3}*{12.1527} & \multirow{3}*{10.1725} \\
 & & \multirow{1}{0.3cm}{125}  &\multirow{1}{0.3cm}{122} &\multirow{1}{0.3cm}{120} &\multirow{1}{0.3cm}{117} &\multirow{1}{0.3cm}{115} &\multirow{1}{0.3cm}{112} &0.1467&12.0059&9.0035&1.1690	& & \\
 & &  \multirow{1}{0.3cm}{105}  &\multirow{1}{0.3cm}{110} &\multirow{1}{0.3cm}{115} &\multirow{1}{0.3cm}{120} &\multirow{1}{0.3cm}{115} &\multirow{1}{0.3cm}{110} &0.0355&12.1171&4.4368&5.7357	& & \\ \cline{2-14}
 & \multirow{3}{0.3cm}{0.5}  & \multirow{1}{0.3cm}{105}  &\multirow{1}{0.3cm}{108} &\multirow{1}{0.3cm}{110} &\multirow{1}{0.3cm}{113} &\multirow{1}{0.3cm}{115} &\multirow{1}{0.3cm}{118} &0.0598&14.8449&4.4038&8.5208	& \multirow{3}*{14.9048} & \multirow{3}*{12.9246} \\
 & &  \multirow{1}{0.3cm}{125}  &\multirow{1}{0.3cm}{122} &\multirow{1}{0.3cm}{120} &\multirow{1}{0.3cm}{117} &\multirow{1}{0.3cm}{115} &\multirow{1}{0.3cm}{112} &0.0803&14.8244&10.5752&2.3494	& & \\
 & & \multirow{1}{0.3cm}{105}  &\multirow{1}{0.3cm}{110} &\multirow{1}{0.3cm}{115} &\multirow{1}{0.3cm}{120} &\multirow{1}{0.3cm}{115} &\multirow{1}{0.3cm}{110} &0.0174&14.8874&4.5862&8.3384	& & \\ \hline

\end{tabular}
\end{subtable}

\begin{subtable}{0.4\textwidth}
\caption{$K=100$, Down-barrier options}
\label{k100dtable_vol}
\renewcommand{\tabcolsep}{1.5mm}
\renewcommand{\arraystretch}{0.9}
\begin{tabular}{ | c | c | c c c c c c | c | c | c | c | c | c | } \hiderowcolors\hline

{$r$} & {$\sigma$} &{${B}_{1}$ } &{${B}_{2}$ }&{${B}_{3}$ }&{${B}_{4}$ }&{${B}_{5}$ }&{${B}_{6}$ } & {DOC} & {DIC} & {DOP} & {DIP} &{Call} & {Put}\\ \hline

\multirow{6}{0.5cm}{0.03}

& \multirow{3}{0.3cm}{0.4}  & \multirow{1}{0.3cm}{75}  &\multirow{1}{0.3cm}{78} &\multirow{1}{0.3cm}{80} &\multirow{1}{0.3cm}{83} &\multirow{1}{0.3cm}{85} &\multirow{1}{0.3cm}{88} &11.1647&0.7579&0.2328&10.2010& \multirow{3}*{11.9226} & \multirow{3}*{10.4338} \\
 & &  95 &92 &90 &87 &85 &83 &5.3671&6.5555&0.1602&10.2735	& & \\
 & &  95 &92 &90 &87 &90 &92 &5.3073&6.6153&0.0219&10.4119      & & \\ \cline{2-14}

 & \multirow{3}{0.3cm}{0.5}  & 75 &78 &80 &83 &85 &88 &13.0059&1.6780&0.1320&13.0631	& \multirow{3}*{14.6840} & \multirow{3}*{13.1952} \\
 & &  95 &92 &90 &87 &85 &83 &5.3917&9.2922&0.0871&13.1081	& & \\
 & &  95 &92 &90 &87 &90 &92 &5.3262&9.3578&0.0109&13.1843	& & \\ \hline

\multirow{6}{0.5cm}{0.04}

& \multirow{3}{0.3cm}{0.4}  & 75 &78 &80 &83 &85 &88 &11.3884&0.7643&0.2310&9.9415 & \multirow{3}*{12.1527} & \multirow{3}*{10.1725} \\
 & &  95 &92 &90 &87 &85 &83 &5.4864&6.6663&0.1587&10.0138	& & \\
 & &  95 &92 &90 &87 &90 &92 &5.4263&6.7264&0.0217&10.1508	& & \\ \cline{2-14}

 & \multirow{3}{0.3cm}{0.5}  & 75 &78 &80 &83 &85 &88 &13.2141&1.6907&0.1313&12.7933	& \multirow{3}*{14.9048} & \multirow{3}*{12.9246} \\
 & &  95 &92 &90 &87 &85 &83 &5.4875&9.4173&0.0865&12.8381	& & \\
 & &  95 &92 &90 &87 &90 &92 &5.4216&9.4831&0.0108&12.9138	& & \\ \hline

\end{tabular}
\end{subtable}

\end{table}

\begin{table}[h] 
\small

\caption{Multi-step option prices with longer maturities; In each cell, the first number is the multi-step option price, the number in the curly brackets is the MC simulated price, the number in the parentheses is the standard error, and the number in the square brackets is the difference between the simulated price and the multi-step formula price. The simulation used 1,000,000 sample paths. RMS relative errors are 0.0034 for up-barrier options and 0.004 for down-barrier options, respectively. Barriers change at every 6 months with levels 105-108-110-113-115-118 for up-barriers and 95-92-90-87-85-83 for down-barriers. } 
\label{ktable_long}

\renewcommand{\tabcolsep}{1.4mm}
\renewcommand{\arraystretch}{0.9}
\medskip
\begin{tabular}{| c | c | c |  c | c | c | c | c | c | c |c |} \hiderowcolors\hline

{$T$} &{$r$} & {$\sigma$} & {UOC} & {UIC} & {UOP} & {UIP}& {DOC} & {DIC} & {DOP} & {DIP}\\ \hline

\multirow{16}{0.5cm}{1}
&\multirow{8}{0.5cm}{0.03}
&\multirow{4}{0.5cm}{0.2}
&0.0341&9.3793&3.1901&3.2678&5.0954&4.3180&0.0422&6.4157\\&&&\{0.0340\}&\{9.3869\}&\{3.1883\}&\{3.2620\}&\{5.0996\}&\{4.3212\}&\{0.0419\}&\{6.4084\}\\&&&(0.0002)&(0.0140)&(0.0061)&(0.0045)&(0.0097)&(0.0063)&(0.0002)&(0.0091)\\&&&[-0.0001]&[0.0075]&[-0.0018]&[-0.0058]&[0.0042]&[0.0032]&[-0.0003]&[-0.0073]\\ \cline{3-11}
&&\multirow{4}{0.5cm}{0.3}
&0.0104&13.2729&3.6705&6.6574&5.1562&8.1271&0.0134&10.3144\\&&&\{0.0104\}&\{13.2562\}&\{3.6791\}&\{6.6697\}&\{5.1426\}&\{8.1241\}&\{0.0134\}&\{10.3353\}\\&&&(0.0001)&(0.0221)&(0.0072)&(0.0084)&(0.0122)&(0.0125)&(0.0001)&(0.0138)\\&&&[0.0000]&[-0.0167]&[0.0086]&[0.0123]&[-0.0137]&[-0.0030]&[0.0000]&[0.0209]\\ \cline{2-11}

&\multirow{8}{0.5cm}{0.04}
&\multirow{4}{0.5cm}{0.2}
&0.0339&9.8911&2.9403&3.0637&5.4184&4.5067&0.0412&5.9628\\&&&\{0.0341\}&\{9.8967\}&\{2.9335\}&\{3.0627\}&\{5.4212\}&\{4.5096\}&\{0.0407\}&\{5.9555\}\\&&&(0.0002)&(0.0143)&(0.0058)&(0.0044)&(0.0105)&(0.0058)&(0.0002)&(0.0089)\\&&&[0.0001]&[0.0055]&[-0.0068]&[-0.0009]&[0.0028]&[0.0029]&[-0.0005]&[-0.0072]\\ \cline{3-11}
&&\multirow{4}{0.5cm}{0.3}
&0.0103&13.7429&3.4688&6.3634&5.3777&8.3756&0.0132&9.8190\\&&&\{0.0103\}&\{13.7623\}&\{3.4684\}&\{6.3679\}&\{5.3871\}&\{8.3855\}&\{0.0133\}&\{9.8231\}\\&&&(0.0001)&(0.0216)&(0.0068)&(0.0082)&(0.0119)&(0.0123)&(0.0001)&(0.0132)\\&&&[0.0000]&[0.0193]&[-0.0003]&[0.0045]&[0.0095]&[0.0099]&[0.0001]&[0.0041]\\ \cline{1-11}

\multirow{16}{0.5cm}{2}
&\multirow{8}{0.5cm}{0.03}
&\multirow{4}{0.5cm}{0.2}
&0.0472&14.0265&3.1845&5.0656&6.1758&7.8978&0.0640&8.1861\\&&&\{0.0472\}&\{14.0240\}&\{3.1795\}&\{5.0625\}&\{6.1840\}&\{7.8872\}&\{0.0642\}&\{8.1778\}\\&&&(0.0003)&(0.0219)&(0.0070)&(0.0078)&(0.0141)&(0.0120)&(0.0004)&(0.0119)\\&&&[0.0000]&[-0.0025]&[-0.0049]&[-0.0032]&[0.0082]&[-0.0106]&[0.0002]&[-0.0083]\\ \cline{3-11}
&&\multirow{4}{0.5cm}{0.3}
&0.0143&19.3683&3.6514&9.9076&5.8655&13.5171&0.0209&13.5381\\&&&\{0.0143\}&\{19.4180\}&\{3.6530\}&\{9.9105\}&\{5.8806\}&\{13.5517\}&\{0.0208\}&\{13.5428\}\\&&&(0.0001)&(0.0339)&(0.0076)&(0.0126)&(0.0165)&(0.0224)&(0.0002)&(0.0170)\\&&&[0.0000]&[0.0497]&[0.0016]&[0.0029]&[0.0151]&[0.0346]&[0.0000]&[0.0046]\\ \cline{2-11}

&\multirow{8}{0.5cm}{0.04}
&\multirow{4}{0.5cm}{0.2}
&0.0466&15.0381&2.8176&4.5787&6.7081&8.3766&0.0612&7.3352\\&&&\{0.0464\}&\{15.0351\}&\{2.8240\}&\{4.5910\}&\{6.7125\}&\{8.3691\}&\{0.0611\}&\{7.3540\}\\&&&(0.0003)&(0.0224)&(0.0064)&(0.0073)&(0.0145)&(0.0125)&(0.0003)&(0.0112)\\&&&[-0.0002]&[-0.0030]&[0.0064]&[0.0123]&[0.0044]&[-0.0075]&[-0.0001]&[0.0188]\\ \cline{3-11}
&&\multirow{4}{0.5cm}{0.3}
&0.0141&20.2657&3.3542&9.2373&6.2044&14.0755&0.0203&12.5711\\&&&\{0.0142\}&\{20.2050\}&\{3.3514\}&\{9.2504\}&\{6.1768\}&\{14.0424\}&\{0.0205\}&\{12.5814\}\\&&&(0.0001)&(0.0342)&(0.0074)&(0.0126)&(0.0168)&(0.0226)&(0.0002)&(0.0170)\\&&&[0.0000]&[-0.0607]&[-0.0028]&[0.0131]&[-0.0275]&[-0.0331]&[0.0001]&[0.0102]\\ \cline{1-11}

\multirow{16}{0.5cm}{3}
&\multirow{8}{0.5cm}{0.03}
&\multirow{4}{0.5cm}{0.2}
&0.0623&17.8372&3.1044&6.1882&6.9786&10.9210&0.0811&9.2116\\&&&\{0.0630\}&\{17.8096\}&\{3.1170\}&\{6.1873\}&\{6.9644\}&\{10.9083\}&\{0.0802\}&\{9.2242\}\\&&&(0.0004)&(0.0291)&(0.0072)&(0.0093)&(0.0170)&(0.0178)&(0.0005)&(0.0129)\\&&&[0.0007]&[-0.0276]&[0.0126]&[-0.0009]&[-0.0142]&[-0.0127]&[-0.0009]&[0.0126]\\ \cline{3-11}
&&\multirow{4}{0.5cm}{0.3}
&0.0188&24.1880&3.5756&12.0243&6.3663&17.8405&0.0266&15.5733\\&&&\{0.0188\}&\{24.1235\}&\{3.5795\}&\{12.0480\}&\{6.3677\}&\{17.7745\}&\{0.0263\}&\{15.6012\}\\&&&(0.0002)&(0.0432)&(0.0085)&(0.0154)&(0.0194)&(0.0306)&(0.0002)&(0.0202)\\&&&[0.0000]&[-0.0645]&[0.0038]&[0.0237]&[0.0014]&[-0.0660]&[-0.0003]&[0.0279]\\ \cline{2-11}

&\multirow{8}{0.5cm}{0.04}
&\multirow{4}{0.5cm}{0.2}
&0.0611&19.3283&2.6563&5.4251&7.6864&11.7030&0.0760&8.0054\\&&&\{0.0613\}&\{19.3620\}&\{2.6576\}&\{5.4079\}&\{7.6970\}&\{11.7262\}&\{0.0757\}&\{7.9898\}\\&&&(0.0004)&(0.0281)&(0.0066)&(0.0083)&(0.0175)&(0.0171)&(0.0005)&(0.0120)\\&&&[0.0002]&[0.0336]&[0.0013]&[-0.0172]&[0.0107]&[0.0232]&[-0.0003]&[-0.0156]\\ \cline{3-11}
&&\multirow{4}{0.5cm}{0.3}
&0.0185&25.4765&3.2085&10.9786&6.7965&18.6985&0.0256&14.1614\\&&&\{0.0188\}&\{25.4819\}&\{3.2167\}&\{10.9842\}&\{6.7734\}&\{18.7273\}&\{0.0255\}&\{14.1754\}\\&&&(0.0002)&(0.0443)&(0.0076)&(0.0145)&(0.0198)&(0.0320)&(0.0002)&(0.0187)\\&&&[0.0003]&[0.0054]&[0.0082]&[0.0056]&[-0.0232]&[0.0288]&[-0.0002]&[0.0140]\\ 
 \hline
\end{tabular}
\end{table}

Table \ref{ktable_long} presents the movements of option prices when the maturity changes. We calculate the multi-step option prices and corresponding MC estimated prices. For easier comparisons, we fix the barrier levels for six steps at 6 month intervals, considering 2 steps for 1 year maturity, 4 steps for 2 year maturity, and 6 steps for 3 year maturity. When the maturity $T$ is 1, up barriers are set to be 105 and then 108 and down barriers are set to be 95 and then 92. When $T=2$, up barriers are $105-108-110-113$ and down barriers are $95-92-90-87$. When $T=3$, up barriers are $105-108-110-113-115-118$ and down barriers are $95-92-90-87-85-83$. As in Tables \ref{k100table_BB_RMS} and \ref{k100dtable_BB}, Table \ref{ktable_long} has 4 numbers in each cell: multi-step prices, MC prices, MC standard errors, and the differences between MC price and the formula price. 
\clearpage

When we investigate multi-step option prices in Table \ref{ktable_long}, we can see that all multi-step barrier options except UOP have increasing prices as the maturity gets longer. UOP prices decrease when the maturity increases, but the price differences are small. Ordinary option prices tend to increase as the maturity increases, but it may not be true with barrier options. UIP price and UOP price sum up to the ordinary put option price when the strike, the barrier, and the expiration are the same. From Table \ref{ktable_long}, UIP prices increase more rapidly than the ordinary put prices as $T$ increases. When the maturity-sensitivity of UIP prices is higher than that of the ordinary put prices, the maturity-sensitivity of UOP prices should be negative. In other words, the decrease of UOP prices in $T$ does not contradict the option pricing theory. 
Another perspective to the decrease of UOP prices would be as follows. As the maturity gets longer, variance of the terminal asset price gets larger, and thus, we generally expect that the ordinary option prices tend to increase as $T$ increases. Therefore, along with more possibility of touching the barrier, in-barrier options prices tend to increase as $T$ increases as we can verify with UIC, UIP, DIC, and DIP prices. However, out-barrier option prices have mixed effects of increased maturity. UOC, DOC, and DOP prices increase with $T$ only slightly. UOP prices even slightly decrease with $T$. In case of UOP, for instance, it is advantageous for the option holder that the underlying prices have a large variance at the expiration because the payoff, $(K-S(T))^+$ is more likely to be large, but on the other hand, it is not advantageous for the underlying prices to be more variable because the probability of being knocked out will be larger. 

Regarding Monte Carlo simulation, RMS relative error for up-barrier options is 0.0034(0.34\%) and that for down-barrier options is 0.004(0.4\%). Overall, SEs and absolute errors are small. In addition, SEs tend to increase as $T$ increases, which is easily anticipated since the standard deviation of the payoff at the expiration increases as $T$ increases and it in turn gives larger SEs. Again, the MC simulation used 1,000,000 sample paths.



\section{Approximation of an arbitrary barrier}

In this section, we present two examples of approximating curved barriers by multi-step barriers. With an arbitrary barrier, the explicit form of the barrier option price is not known. One can try numerical calculation with Monte Carlo methods; however, often times, Monte Carlo methods do not provide a fast and reliable result. Thus, we would like to suggest approximating the prices of curved-barrier options with muti-step barrier options' prices. In Examples \ref{ex4} and \ref{ex5}, we approximate the probabilities that are the basis for calculating curved-barrier option prices by the formula given in Proposition \ref{prop1}. At the end of this section, we also compare the approximated option prices by multi-step barriers with the option prices computed from formulas given in \cite{kunitomo1992} when the barrier is of the exponential form.

\begin{example}\label{ex4}
Let $\{X(t): 0 \leq t \leq T \}$ denote a Brownian motion with drift $\mu$ and diffusion coefficient $\sigma$ and $M(s,t)$ denote the maximum of $X$ in the time interval from $s$ to $t$. For a fixed time point $t$, the probability that $X(t)$ is less than $\mu t + z_{\alpha} \sigma \sqrt{t}$ is $1-\alpha$ where $z_{\alpha}$ is the $100(1-\alpha)$-percentile of the standard normal distribution; for instance, such probability is 0.95 if $\mu t +1.645 \sigma \sqrt{t}$ is used. Suppose we set 
$$g(t)=\mu t + z_{\alpha} \sigma \sqrt{t}$$
as the upper barrier and would like to calculate the probability
\begin{equation}
\label{ex4eqn1}
Pr(X(t) \leq g(t), ~0 \leq t \leq T).
\end{equation}
Then the probability can be approximated by (\ref{prop1eqn1}) with $m_i=g(t_i)$, $m_i=g(t_{i-1})$, or $m_i=(g(t_{i-1})+g(t_i))/2$ for $0=t_0<t_1 <\cdots<t_n=T$. As in Corollary \ref{coro1}, $x_i$'s are to be set as $m_i \wedge m_{i+1}$ for $i=0, \cdots, n-1$ and $x_n=m_n$ since there exist no icicles. For a sufficiently large $n$, we believe that (\ref{prop1eqn1}) will give us a good approximation for (\ref{ex4eqn1}). Table \ref{ex4table} provides the approximated values of (\ref{ex4eqn1}) with $m_i=(g(t_{i-1})+g(t_i))/2$, $n=10$, $t_1=1/20$, and $t_2=2/20$, $\cdots$, $t_n=T=10/20$. As $1-\alpha$ increases, the probability that $X(t)$ stays under the barrier gets increased, which complies with our intuition. We also observe that the probability of not hitting the barrier stays almost the same when we change the values of $\mu$ and $\sigma$. It is a natural result because $X(t) \leq g(t)$ is equivalent to $\frac{X(t)-\mu t}{\sigma \sqrt{t}} \leq z_{\alpha}$ and $\frac{X(t)-\mu t}{\sigma \sqrt{t}}$ has the same distribution regardless of the values of $\mu$ and $\sigma$.

\end{example}

\begin{table}[hbt!]
\caption{Approximation of $Pr(X(t) \leq \mu t + z_{\alpha} \sigma \sqrt{t}, ~0 \leq t \leq T)$ in Example \ref{ex4}, with $n=10$, $T=0.5$, and equally-spaced $t_i$'s}
\renewcommand{\tabcolsep}{1.6mm}
\renewcommand{\arraystretch}{1.1}
\begin{center}
\begin{tabular}{|c|c|c|c|c|c|c|c|c|c|} \hiderowcolors\cline{1-10}

\multicolumn{1}{|c|}{} & \multicolumn{3}{c|}{$\mu=0.01$} & \multicolumn{3}{|c|}{$\mu=0.02$}& \multicolumn{3}{|c|}{$\mu=0.03$}\\ \cline{1-10}

$1-\alpha$ & $\sigma=0.1$ & $\sigma=0.2$& $\sigma=0.3$ & $\sigma=0.1$ & $\sigma=0.2$ & $\sigma=0.3$ & $\sigma=0.1$ & $\sigma=0.2$ & $\sigma=0.3$ \\ \cline{1-10}
0.90 & 0.3692&	0.3693&	0.3693&	0.3689&	0.3692&	0.3693&	0.3686&	0.3691&	0.3692 \\ \hline
0.95 & 0.5074&	0.5079&	0.5081&	0.5062&	0.5074&	0.5077&	0.5050&	0.5068&	0.5074 \\ \hline
0.99 & 0.7230&	0.7239&	0.7242&	0.7211&	0.7230&	0.7236&	0.7192&	0.7220&	0.7230 \\ \hline
\end{tabular}
\end{center}
\label{ex4table}
\end{table}

\begin{example}\label{ex5}
Under the Black-Scholes model, 
$$S(t) =S(0) e^{X(t)}, \qquad t \geq 0,$$
where $\{X(t): t \geq 0\}$ is a Brownian motion with the drift parameter $\mu$ and the diffusion parameter $\sigma$. Suppose we are interested in the event that the underlying asset price does not go beyond a linear barrier; in other words, we are interested in the event of 
$$\{ S(t) \leq B(t)=C_0+C_1 t, ~0 \leq t \leq T\},$$
for some positive constants $C_0$ and $C_1$.
Then this event is the same as 
$$\{X(t) \leq \ln \left( \frac{C_0 +C_1 t}{S(0)} \right), 0 \leq t \leq T\},$$
which means that the upper barrier is a curved one. Again, by setting  $m_i=\ln \left(\frac{C_0+C_1 t_i}{S(0)}\right)$, $m_i=\ln \left(\frac{C_0+C_1 t_{i-1}}{S(0)}\right)$, or $m_i=\left(\ln \left(\frac{C_0+C_1 t_i}{S(0)}\right)+\ln \left(\frac{C_0+C_1 t_{i-1}}{S(0)}\right)\right)/2$ for $0=t_0<t_1 <\cdots<t_n=T$, (\ref{prop1eqn1}) will approximate the probability of the event of interest. It is the up-barrier case without icicles, so $x_i$'s are to be set as $m_i \wedge m_{i+1}$ for $i=0, \cdots, n-1$ and $x_n$ as $m_n$. In Table \ref{ex5table}, we calculate the approximation of the probability that $S(t)$ is under $C_0+C_1 t$ for $0 \leq t \leq T$ using (\ref{prop1eqn1}) with ${m}_{i}= (\ln(\frac{{C}_{0}+{C}_{1}{t}_{i}}{S(0)})+\ln(\frac{{C}_{0}+{C}_{1}{t}_{i-1}}{S(0)}))/2$, $\mu=0.01$, and $\sigma=0.2$. $S(0)$, $n$, and $T$ are assumed to be 100, 10, and 0.5, respectively and $t_i$'s are equally spaced. As easily expected, as the intercept $C_0$ and/or the slope $C_1$ of the linear barrier increase, the probability of up-barrier not being hit increases. 

\end{example}

\begin{table}[hbt!]
\caption{Approximation of $Pr(S(t) \leq C_0+C_1 t, ~0 \leq t \leq T)$ in Example \ref{ex5}, with $S(0)=100$, $n=10$, $T=0.5$, and equally-spaced $t_i$'s}
\renewcommand{\tabcolsep}{2mm}
\renewcommand{\arraystretch}{1.1}
\begin{center}


\begin{tabular}{|c|c|c|c|c|c|} \hiderowcolors\cline{1-6}
\diagbox{${C}_{0}/{S(0)}$}{${C}_{1}/{S(0)}$}&{0}&{0.25}&{0.5}&{0.75}&{1}\\ \hline
1.025&	0.1333&	0.2867&	0.4447&	0.5760&	0.6757\\ \hline
1.050&	0.2611&	0.4744&	0.6530&	0.7743&	0.8501\\ \hline
1.075&	0.3799&	0.6178&	0.7874&	0.8849&	0.9359\\ \hline
1.100&	0.4877&	0.7233&	0.8688&	0.9409&	0.9730\\ \hline
\end{tabular}
\end{center}
\label{ex5table}
\end{table}

In order to illustrate the performance of the approximation of curved barrier option prices by multi-step barrier options, we compare the numerical results of multi-step barrier option prices with those by \cite{kunitomo1992} in Table \ref{Kunitomo8}. \cite{kunitomo1992} provides a European option valuation method with double(up and down) exponential barriers of the form, $A_2 \exp(\delta_2 t)$ and $A_1 \exp (\delta_1 t)$. The general pricing formula involves an infinite series, but they found an explicit option pricing formula for single exponential barrier options. Corollaries 4.1 and 4.2 in their paper calculate the down-and-out call option prices with the curved barrier, $A_2 \exp(\delta_2t)$ and the up-and-out put option prices with the barrier, $A_1 \exp(\delta_1 t)$, respectively. Table \ref{Kuni_UOP8} contains up-and-out put option prices when the up-barrier is $S(t) \leq A_1 \exp(\delta_1 t), 0 \leq t \leq 0.5$ with $A_1 =105$ or $110$ and $\delta_1=0.1$ or $0.2$. Table \ref{Kuni_DOC8} contains down-and-out call option prices when the down-barrier is $S(t) \geq A_2 \exp(\delta_2 t), 0 \leq t \leq 0.5$ with $A_2=95$ or $90$ and $\delta_2=-0.1$ or $-0.2$. To approximate the curved barrier with a multi-step barrier, we set the subintervals as $t_0=0,$ $t_1=1/24$, $\cdots$, $t_4=4/24$, $t_5=3/12$, $\cdots$, $t_8=T= 6/12$. Note that we have tighter subintervals in the first third, but we can use different sets of subintervals. From Table \ref{Kunitomo8}, we can see that the multi-step approximation works quite well with $n=8$ subintervals. Multi-step barriers can approximate not only exponential barriers but also any type of barriers. 
 
 	\begin{table}[]
\small
\centering
\caption{Exponential barrier option prices by \cite{kunitomo1992} and multi-step approximation; multi-step columns present option prices approximated by multi-step formula and K\&I columns present the option prices computed from Corollaries 4.1 and 4.2 in \cite{kunitomo1992}. Multi-step barriers are set as $B_i=(c(t_{i-1})+c(t_{i}))/2)$ where $c(t)$ is a curved knockout boundary given below and $t_0=0,$ $t_1=1/24$, $\cdots$, $t_4=4/24$, $t_5=3/12$, $\cdots$, $t_8=T= 6/12$. $S(0)=100$.}
\label{Kunitomo8}

\begin{subtable}{0.8\textwidth}
\caption{Up-and-out put option prices with $c(t)=A_1 \exp(\delta_1t)$ for $0 \leq t \leq T$}

\label{Kuni_UOP8}
\renewcommand{\tabcolsep}{1.6mm}
\renewcommand{\arraystretch}{1}
\begin{tabular}{|c|c|c|c|c|c||c|c||c|c|} \hiderowcolors \cline{5-10}
                      \multicolumn{4}{c|}{}& \multicolumn{2}{c||}{$K=90$}  & \multicolumn{2}{c||}{$K=100$} & \multicolumn{2}{c|}{$K=110$} \\ \cline{1-10}
$r$                     & $\sigma$                & $A_1$                    & $\delta_1$  & K\&I    & multi-step      & K\&I    & multi-step       & K\&I    & multi-step       \\ \cline{1-10}
\multirow{8}{*}{0.03} & \multirow{4}{*}{0.2} & \multirow{2}{*}{105} & 0.1 & 1.1605      & 1.1597   & 3.4520      & 3.4486    & 6.7481      & 6.7377    \\
                      &                      &                      & 0.2 & 1.2252      & 1.2216   & 3.7616      & 3.7467    & 7.6629      & 7.6234    \\ \cline{3-10}
                      &                      & \multirow{2}{*}{110} & 0.1 & 1.4127      & 1.4126   & 4.5558      & 4.5550    & 9.6632      & 9.6597    \\
                      &                      &                      & 0.2 & 1.4304      & 1.4298   & 4.6795      & 4.6758    & 10.1602     & 10.1461   \\ \cline{2-10}
                      & \multirow{4}{*}{0.3} & \multirow{2}{*}{105} & 0.1 & 2.0823      & 2.0878   & 3.9928      & 4.0033    & 6.2574      & 6.2722    \\
                      &                      &                      & 0.2 & 2.2321      & 2.2397   & 4.3697      & 4.3832    & 7.0120      & 7.0310    \\ \cline{3-10}
                      &                      & \multirow{2}{*}{110} & 0.1 & 2.9978      & 2.9969   & 6.0633      & 6.0613    & 9.9626      & 9.9585    \\
                      &                      &                      & 0.2 & 3.1014      & 3.0977   & 6.3682      & 6.3587    & 10.6573     & 10.6379   \\ \cline{1-10}
\multirow{8}{*}{0.04} & \multirow{4}{*}{0.2} & \multirow{2}{*}{105} & 0.1 & 1.0841      & 1.0833   & 3.2739      & 3.2707    & 6.4578      & 6.4477    \\
                      &                      &                      & 0.2 & 1.1450      & 1.1416   & 3.5711      & 3.5568    & 7.3498      & 7.3114    \\ \cline{3-10}
                      &                      & \multirow{2}{*}{110} & 0.1 & 1.3213      & 1.3212   & 4.3316      & 4.3309    & 9.2895      & 9.2861    \\
                      &                      &                      & 0.2 & 1.3380      & 1.3374   & 4.4512      & 4.4476    & 9.7780      & 9.7642    \\ \cline{2-10}
                      & \multirow{4}{*}{0.3} & \multirow{2}{*}{105} & 0.1 & 1.9949      & 2.0002   & 3.8478      & 3.8579    & 6.0524      & 6.0667    \\
                      &                      &                      & 0.2 & 2.1393      & 2.1466   & 4.2139      & 4.2268    & 6.7903      & 6.8086    \\ \cline{3-10}
                      &                      & \multirow{2}{*}{110} & 0.1 & 2.8756      & 2.8748   & 5.8538      & 5.8519    & 9.6612      & 9.6572    \\
                      &                      &                      & 0.2 & 2.9758      & 2.9721   & 6.1509      & 6.1417    & 10.3431     & 10.3241  \\ \hline

\end{tabular}
\end{subtable}

\bigskip

\begin{subtable}{0.8\textwidth}
\caption{Down-and-out call option prices with $c(t)=A_2 \exp(\delta_2t)$ for $0 \leq t \leq T$}
\label{Kuni_DOC8}
\renewcommand{\tabcolsep}{1.6mm}
\renewcommand{\arraystretch}{1}
\begin{tabular}{|c|c|c|c|c|c||c|c||c|c|} \cline{5-10}
                      \multicolumn{4}{c|}{}& \multicolumn{2}{c||}{$K=90$}  & \multicolumn{2}{c||}{$K=100$} & \multicolumn{2}{c|}{$K=110$} \\ \cline{1-10}
$r$            & $\sigma$                & $A_2$                    & $\delta_2$  & K\&I    & multi-step      & K\&I    & multi-step       & K\&I    & multi-step      \\ \hline
\multirow{8}{*}{0.03} & \multirow{4}{*}{0.2} & \multirow{2}{*}{95} & $-0.1$ & 8.3666   & 8.3531  & 4.7296   & 4.7246 & 2.1256   & 2.1238 \\
                      &                      &                     & $-0.2$ & 9.3694   & 9.3215  & 5.1010   & 5.0808 & 2.2356   & 2.2288 \\ \cline{3-10}
                      &                      & \multirow{2}{*}{90} & $-0.1$ & 11.7384  & 11.7354 & 6.1027   & 6.1021 & 2.5581   & 2.5580 \\
                      &                      &                     & $-0.2$ & 12.1728  & 12.1611 & 6.2126   & 6.2096 & 2.5801   & 2.5795 \\ \cline{2-10}
                      & \multirow{4}{*}{0.3} & \multirow{2}{*}{95} & $-0.1$ & 7.2767   & 7.2884  & 5.0570   & 5.0665 & 3.1979   & 3.2041 \\
                      &                      &                     & $-0.2$ & 8.0958   & 8.1061  & 5.4979   & 5.5072 & 3.4179   & 3.4247 \\ \cline{3-10}
                      &                      & \multirow{2}{*}{90} & $-0.1$ & 11.6246  & 11.6207 & 7.6586   & 7.6566 & 4.6069   & 4.6059 \\
                      &                      &                     & $-0.2$ & 12.3043  & 12.2854 & 7.9698   & 7.9602 & 4.7389   & 4.7343 \\ \cline{1-10}
\multirow{8}{*}{0.04} & \multirow{4}{*}{0.2} & \multirow{2}{*}{95} & $-0.1$ & 8.6552   & 8.6415  & 4.9367   & 4.9315 & 2.2470   & 2.2451 \\
                      &                      &                     & $-0.2$ & 9.6714   & 9.6227  & 5.3191   & 5.2983 & 2.3622   & 2.3550 \\ \cline{3-10}
                      &                      & \multirow{2}{*}{90} & $-0.1$ & 12.0836  & 12.0806 & 6.3528   & 6.3523 & 2.7004   & 2.7002 \\
                      &                      &                     & $-0.2$ & 12.5193  & 12.5075 & 6.4651   & 6.4621 & 2.7233   & 2.7226 \\ \cline{2-10}
                      & \multirow{4}{*}{0.3} & \multirow{2}{*}{95} & $-0.1$ & 7.4677   & 7.4798  & 5.2095   & 5.2193 & 3.3109   & 3.3172 \\
                      &                      &                     & $-0.2$ & 8.2987   & 8.3094  & 5.6599   & 5.6696 & 3.5372   & 3.5443 \\ \cline{3-10}
                      &                      & \multirow{2}{*}{90} & $-0.1$ & 11.8948  & 11.8908 & 7.8733   & 7.8713 & 4.7626   & 4.7616 \\
                      &                      &                     & $-0.2$ & 12.5811  & 12.5619 & 8.1900   & 8.1802 & 4.8981   & 4.8933\\ \hline
\end{tabular}
\end{subtable}
\end{table}

\section{Conclusion}

In this paper, we derived the general and explicit pricing formulas for the multi-step barrier options with an arbitrary number of steps and levels with or without icicles under the Black-Scholes model. In the meantime, we introduced the multi-step reflection principle which reveals the mechanism of repeated applications of the ordinary reflection principle of Brownian motion. Using Esscher transform and the multi-step reflection principle, we were able to find the joint distribution of the Brownian motion at the end of subintervals of time and their partial maximums, which is directly used to obtain the option pricing formula. We conducted numerical studies on eight types of barrier options with different types of barriers under different settings of volatilities, interest rates, strike prices, and maturities, and considered examples of approximating curved barriers by multi-step barriers with many steps. Monte Carlo simulated prices and their standard errors are also provided for different settings of input variables.

The novel contributions of this paper compared to \cite{lee2019a} are that the generalized multi-step reflection principle is found and that any finite number of steps and barrier levels can be handled by the proposed barrier structure. This study is important in the respect that we found the general pricing formulas that have not explicitly expressed so far in the literature. 
Multi-step barrier scheme equipped with the general pricing formula can be useful in the analysis of various financial products. Since the multi-step reflection principle applies to cases with any finite number of steps and we have the general pricing formula, proper sequences of multi-step barriers could approximate arbitrary barriers including curved barriers at a desired precision. We confirmed the good approximating performance by comparing the approximated prices with prices computed by the formula given in \cite{kunitomo1992} when the barrier is of the exponential form.  Another application would be, for example, the deposit insurance scheme where a multi-step barrier option is activated when the assets and liabilities hit the regulatory barriers. These applications are topics of future research. We would like to investigate the efficiency of the approximation of curved-barrier option prices with multi-step barrier option prices in several aspects, extend the results to cases with multi-assets, and find an explicit formula for the premium of deposit insurances in the future.

\medskip

{\bf Data Availability Statement:} Data sharing is not applicable to this article as no new data were created or analyzed in this study.


\bibliography{ref_mstep_rev3}

\section*{Appendix}
\begin{appendices}

\section{Proof of Theorem \ref{thm0}}
\label{proofthm0}

We assume the cardinality of $J$ is $g$ which means that there are $g$ barriers with possibly different levels in $n$ subperiods, $(t_{i-1}, t_i)$ for $i=1, \cdots, n$. Let us denote $J= \{ i_1, i_2, \cdots, i_g\}$ with $1 \leq i_1 < i_2 < \cdots < i_g \leq n$, and divide the set of time points $t_1, \cdots, t_n$ into $g+1$ groups in such a way that the index sets are 

\begin{align*}
J_1&=\{1, \cdots, i_1 -1\}\\
J_2&=\{i_1, \cdots, i_2 -1\}\\
  \vdots &\hspace{1.5cm} \vdots\\
J_g&=\{i_{g-1}, \cdots, i_g -1\}\\
J_{g+1}&=\{i_g, \cdots, n\}.
\end{align*}

We assume that $J_1=\varnothing$ when the first subinterval has a horizontal barrier such that $i_1=1$. Then
\begin{align}
Pr(&\cap_{i=1}^n \{ X^0(t_i) \leq x_i \} , \cap_{\{k \in J\}} \{M^0(t_{k-1}, t_k) > m_k\}) \label{original0} \\
   &=Pr( \cap_{\{i \in J_1\}} \{X^0(t_i) \leq x_i\}, M^0(t_{i_1-1}, t_{i_1})>m_{i_1},  \cap_{\{i \in J_2\}} \{X^0(t_i) \leq x_i\}, \nonumber \\
   & \qquad      M^0(t_{i_2-1}, t_{i_2})>m_{i_2}, \cdots, \cap_{\{i \in J_g\}} \{X^0(t_i) \leq x_i\}, M^0(t_{i_g-1}, t_{i_g})>m_{i_g}, \nonumber \\
   &\qquad       \cap_{\{i \in J_{g+1}\}} \{X^0(t_i) \leq x_i\} ). \label{reflect0}
\end{align}

Let us use the mathematical induction for showing (\ref{eqn3_30}). Equation (\ref{eqn3_30}) is trivial with $g=0$. Suppose $g=1$ and the barrier is in $(t_{i_1-1}, t_{i_1})$. By the ordinary reflection principle, $X^0(t)$ for $0 < t < T$ has the same distribution as
$$\left\{\begin{array}{lc} X^0(t), & 0 < t<\tau_1,\\ 2m_{i_1} -X^0(t), & \tau_1<t < T,\end{array}\right.$$
where $m_{i_1}$ is the barrier level and $\tau_1 \in (t_{i_1-1}, t_{i_1})$ is the time point when the barrier is hit. 
Thus for $J_1=\{1, \cdots, i_1 -1\}$ and $J_2=\{i_1, \cdots, n\}$,
\begin{align*}
Pr(&\cap_{i=1}^n \{ X^0(t_i) \leq x_i \} , M^0(t_{i_1-1}, t_{i_1}) > m_{i_1}) \\
   &=Pr( \cap_{\{i \in J_1\}} \{X^0(t_i) \leq x_i\}, \cap_{\{i \in J_2\}} \{2m_{i_1}-X^0(t_i) \leq x_i\}).
\end{align*}
Since the drift parameter of $X^0$ is zero, it turns to
\begin{align*}
   Pr(& \cap_{\{i \in J_1\}} \{X^0(t_i) \leq x_i\}, \cap_{\{i \in J_2\}} \{2m_{i_1}-X^0(t_i) \leq x_i\}) \\
   &=Pr(\cap_{\{i \in J_1\}} \{-X^0(t_i) \leq x_i\}, \cap_{\{i \in J_2\}} \{X^0(t_i)+2m_{i_1} \leq x_i\}) \\
   &=Pr(\cap_{i=1}^n \{ s_iX^0(t_i) +2m[i] \leq x_i \}).
\end{align*} Note that $m[i]=m_i-m[i-1]$ for $i \in J$, $m[i]=m[i-1]$ for $i \notin J$, and $m[0]=0$. Also, when $i_1=1$, $J_1 =\varnothing$.

Now assume (\ref{eqn3_30}) holds when there are $g$ horizontal barriers and see if the same equation holds when one more horizontal barrier is added. We assume that after $g$ reflections of the Brownian path, $X^0(t)$ has the same distribution as
$$\left\{\begin{array}{lc} X^0(t), & 0 < t<\tau_1,\\ -X^0(t)+2m_{i_1}, & \tau_1<t < \tau_2, \\
X^0(t)+2(m_{i_2}-m_{i_1}), & \tau_2 <t <\tau_3, \\
 \vdots&\vdots\\
(-1)^{g-1} X^0 (t) + 2(m_{i_{g-1}}-m_{i_{g-2}}+ \cdots +(-1)^g m_{i_1}),& \tau_{g-1}<t<\tau_g,\\
(-1)^{g} X^0 (t) + 2(m_{i_{g}}-m_{i_{g-1}}+ \cdots +(-1)^{g+1} m_{i_1}),& \tau_g<t<T, \end{array}\right.$$
where $\tau_i$ is the time point where the $i$th reflection is occurred. In other words, we assume
 \begin{align*}
Pr( &\cap_{i=1}^n \{X^0(t_i) \leq x_i\}, \cap_{k \in J} \{M^0(t_{k-1}, t_{k})>m_{k}\})\\
 &=Pr( \cap_{\{i \in J_1\}} \{X^0(t_i) \leq x_i\}, M^0(t_{i_1-1}, t_{i_1})>m_{i_1},  \cap_{\{i \in J_2\}} \{X^0(t_i) \leq x_i\}, \\
   & \qquad      M^0(t_{i_2-1}, t_{i_2})>m_{i_2}, \cdots, \cap_{\{i \in J_g\}} \{X^0(t_i) \leq x_i\}, M^0(t_{i_g-1}, t_{i_g})>m_{i_g}, \\
   &\qquad       \cap_{\{i \in J_{g+1}\}} \{X^0(t_i) \leq x_i\} )\\
   &= Pr(  \cap_{\{i \in J_1\}} \{X^0(t_i) \leq x_i\},  \cap_{\{i \in J_2\}} \{-X^0(t_i) +2m_{i_1}\leq x_i\},\\
   &\qquad  \cap_{\{i \in J_3\}} \{X^0(t_i) +2(m_{i_2}-m_{i_1})\leq x_i\}, \cdots, \\
   &\qquad  \cap_{\{i \in J_g\}} \{(-1)^{g-1} X^0(t_i) +2(m_{i_{g-1}}-m_{i_{g-2}}+ \cdots+(-1)^g m_{i_1}) \leq x_i\},\\
      &\qquad  \cap_{\{i \in J_{g+1}\}} \{(-1)^g X^0(t_i) +2(m_{i_{g}}-m_{i_{g-1}}+ \cdots+(-1)^{g+1} m_{i_1}) \leq x_i\}.
      \end{align*}
Now, consider adding one more barrier so that $J$ becomes $\{1, 2, \cdots, g+1\}$. This time, the set of time points $t_1, \cdots, t_n$ is divided into $g+2$ groups whose index sets are 
\begin{align*}
J_1&=\{1, \cdots, i_1 -1\}\\
J_2&=\{i_1, \cdots, i_2 -1\}\\
   \vdots  &\hspace{1.5cm} \vdots \\
J_{g+1}&=\{i_g, \cdots, i_{g+1}-1\}\\
J_{g+2}&=\{i_{g+1}, \cdots, n\}.
\end{align*}
With these $g+2$ groups,
\begin{align}
Pr( &\cap_{i=1}^n \{X^0(t_i) \leq x_i\}, \cap_{k \in J} \{M^0(t_{k-1}, t_{k})>m_{k}\}) \nonumber \\
   &=Pr( \cap_{i=1}^{i_{g+1}-1} \{X^0(t_i) \leq x_i\},  \cap_{k \in \{i_1, \cdots, i_g\}} \{M^0(t_{k-1}, t_{k})>m_{k}\}, \nonumber\\
   &\qquad \{M^0(t_{i_{g+1}-1}, t_{i_{g+1}}) > m_{i_{g+1}}\}, \cap_{i=i_{g+1}}^{n} \{X^0(t_i) \leq x_i\}) \nonumber\\
    &= Pr(\cap_{\{i \in J_1\}} \{X^0(t_i) \leq x_i\},  \cap_{\{i \in J_2\}} \{-X^0(t_i) +2m_{i_1}\leq x_i\},  \cdots, \nonumber \\
    &\qquad \cap_{\{i \in J_{g+1}\}} \{(-1)^g X^0(t_i) +2(m_{i_g}-m_{i_{g-1}}+ \cdots+(-1)^{g+1} m_{i_1}) \leq x_i\}, \nonumber \\
   &\qquad \cap_{\{i \in J_{g+2}\}}  \{ (-1)^g X^0(t_i) +2(m_{i_g}-m_{i_{g-1}}+ \cdots+(-1)^{g+1} m_{i_1}) \leq x_i\}, \nonumber \\
   & \max_{i_{g+1}-1 < \tau < i_{g+1}} \{(-1)^g X^0(\tau) +2(m_{i_g}-m_{i_{g-1}}+ \cdots+(-1)^{g+1} m_{i_1}\} > m_{i_{g+1}}) \nonumber\\
   &=Pr(\cap_{\{i \in J_1\}} \{X^0(t_i) \leq x_i\},  \cap_{\{i \in J_2\}} \{-X^0(t_i) +2m_{i_1}\leq x_i\}, \nonumber\\
   &\qquad \cdots, \cap_{\{i \in J_{g+1}\}} \{(-1)^g X^0(t_i) +2(m_{i_g}-m_{i_{g-1}}+ \cdots+(-1)^{g+1} m_{i_1}) \leq x_i\}, \nonumber\\
   &\qquad \cap_{\{i \in J_{g+2}\}} \{2m_{g+1}-[(-1)^g X^0(t_i) +2(m_{i_g}-m_{i_{g-1}}+ \cdots+(-1)^{g+1} m_{i_1})] \leq x_i\}) \nonumber\\
   &=Pr(\cap_{\{i \in J_1\}} \{X^0(t_i) \leq x_i\},  \cap_{\{i \in J_2\}} \{-X^0(t_i) +2m_{i_1}\leq x_i\}, \nonumber\\
   &\qquad \cdots, \cap_{\{i \in J_{g+1}\}} \{(-1)^g X^0(t_i) +2(m_{i_g}-m_{i_{g-1}}+ \cdots+(-1)^{g+1} m_{i_1}) \leq x_i\}, \nonumber\\
   &\qquad \cap_{\{i \in J_{g+2}\}} \{(-1)^{g+1} X^0(t_i) +2(m_{i_{g+1}}-m_{i_g}+ \cdots+(-1)^{g+2} m_{i_1}) \leq x_i\}). \label{eveng}
\end{align}
If $g$ is even, we replace $X^0(t)$ by $-X^0(t)$ and the above probabilities do not change since $X^0$ has the zero drift. Then we can easily check that 
$$(\ref{eveng})=Pr(\cap_{i=1}^n \{s_i X^0(t_i) +2m[i] \leq x_i \}),$$
which is the desired result. 

\section{Proof of Theorem \ref{thm1}}
\label{proofthm1}

To apply the multi-step reflection principle to a Brownian motion with nonzero drift, we use the method of Esscher transform 
as follows. When $\{X(t): 0 \leq t \leq T \}$ is a Brownian motion with drift $\mu$ and diffusion coefficient $\sigma$ and $M(s,t)$ is the maximum of $X$ in the time interval $(s,t)$, 
\begin{align*}
Pr(&\cap_{i=1}^n \{X(t_i) \leq x_i\}, \cap_{\{i \in J\}} \{M(t_{i-1}, t_i) > m_i\})\\
&=E\left[e^{\frac{\mu}{\sigma^2} X(T)-\frac{\mu}{\sigma^2} X(T)}I\left( \cap_{i=1}^n \{X(t_i) \leq x_i\}, \cap_{\{i \in J\}} \{M(t_{i-1}, t_i) > m_i\} \right) \right]\\
&=E[e^{-\frac{\mu}{\sigma^2} X(T)}] E\left[e^{\frac{\mu}{\sigma^2} X(T)}I\left( \cap_{i=1}^n \{X(t_i) \leq x_i\}, \cap_{\{i \in J\}} \{M(t_{i-1}, t_i) > m_i\} \right); -\frac{\mu}{\sigma^2} \right]. 
\end{align*}
The second equality above uses (\ref{expyh}) with $h=-\frac{\mu}{\sigma^2}$.
As seen in Section \ref{prelim}, $X$ becomes a Brownian motion with drift $\mu+h\sigma^2$ and diffusion $\sigma$ under the Esscher measure of the parameter $h$. Since $h=-\frac{\mu}{\sigma^2}$, the drift of $X$ under the Esscher measure is zero. Thus, we can write the above probability as
\begin{align*}
E[&e^{-\frac{\mu}{\sigma^2} X(T)}] E\left[e^{\frac{\mu}{\sigma^2} X^0(T)}I\left( \cap_{i=1}^n \{X^0(t_i) \leq x_i\}, \cap_{\{i \in J\}} \{M^0(t_{i-1}, t_i) > m_i\} \right) \right]\\
&=E[e^{-\frac{\mu}{\sigma^2} X(T)}]\int_{x \leq x_n} e^{\frac{\mu}{\sigma^2} x} dPr(\cap_{i=1}^{n-1} \{X^0(t_i) \leq x_i\}, X_0(T) \leq x, \cap_{\{i \in J\}} \{M^0(t_{i-1}, t_i) > m_i\} )\\
&=E[e^{-\frac{\mu}{\sigma^2} X(T)}]\int_{x \leq x_n}  e^{\frac{\mu}{\sigma^2} x} dPr(\cap_{i=1}^{n-1} \{s_iX^0(t_i) +2m[i] \leq x_i\}, s_nX^0(T)+2m[n] \leq x ) \\
&=E[e^{-\frac{\mu}{\sigma^2} X(T)}]E[e^{\frac{\mu}{\sigma^2} (s_n X^0(T)+2m[n])} I\left( \cap_{i=1}^{n} \{s_iX^0(t_i) +2m[i] \leq x_i\} \right)]\\
&=E[e^{-\frac{\mu}{\sigma^2} X(T)}]E[e^{\frac{\mu}{\sigma^2} (X^0(T)+2m[n])} I\left( \cap_{i=1}^{n} \{s_iX^0(t_i) +2m[i] \leq x_i\} \right)],
\end{align*}
since $s_n=1$. 
Using $E[e^{-\frac{\mu}{\sigma^2} X(T)}] \times E[e^{\frac{\mu}{\sigma^2} X^0(T)}]=1$ and (\ref{expyh}) again, we get
\begin{align}
E[&e^{-\frac{\mu}{\sigma^2} X(T)}]E[e^{\frac{\mu}{\sigma^2} (X^0(T)+2m[n])} I\left( \cap_{i=1}^{n} \{s_iX^0(t_i) +2m[i] \leq x_i\} \right)]\nonumber \\
  &= E[e^{-\frac{\mu}{\sigma^2} X(T)}] \times E[e^{\frac{\mu}{\sigma^2} X^0(T)}] e^{\frac{2\mu}{\sigma^2}m[n]} E[I\left( \cap_{i=1}^{n} \{s_iX^0(t_i) +2m[i] \leq x_i\}\right); \frac{\mu}{\sigma^2}]\nonumber\\
  &=e^{\frac{2\mu}{\sigma^2}m[n]} E[I\left( \cap_{i=1}^{n} \{s_iX^0(t_i) +2m[i] \leq x_i\}\right); \frac{\mu}{\sigma^2}]\nonumber\\
  &=e^{\frac{2\mu}{\sigma^2}m[n]} Pr(\cap_{i=1}^{n} \{s_i X(t_i) +2m[i] \leq x_i\} ). \label{thm0nonzero}
    \end{align}
For the last identity, we used the fact that $X^0$ is a Brownian motion with drift $\mu$ and diffusion coefficient $\sigma$ under the Esscher measure of the parameter $\frac{\mu}{\sigma^2}$. 

\end{appendices}

\end{document}